\documentclass{article}

\usepackage{microtype}
\usepackage{graphicx}
\usepackage{subfigure}
\usepackage{booktabs} % for professional tables
\usepackage{array}

\usepackage{hyperref}
\usepackage{multirow}
\usepackage[table]{xcolor}
\usepackage{adjustbox}

\usepackage{algorithm}
\usepackage{algpseudocode}

\usepackage[accepted]{icml2025}

\usepackage{amsmath}
\usepackage{amssymb}
\usepackage{mathtools}
\usepackage{amsthm}

\usepackage[capitalize,noabbrev]{cleveref}
\newcommand{\method}{ArrayDPS}
\definecolor{top1}{rgb}{0.8, 0.5, 0.5}    % Very Light Dark Red
\definecolor{top2}{rgb}{0.95, 0.7, 0.7}   % Lighter Medium Red
\definecolor{top3}{rgb}{0.98, 0.9, 0.9}  % Extremely Light Red
\theoremstyle{plain}
\newtheorem{theorem}{Theorem}[section]

\theoremstyle{definition}

\theoremstyle{remark}

\usepackage[textsize=tiny]{todonotes}

\usepackage{soul}

\icmltitlerunning{ArrayDPS: Unsupervised Blind Speech Separation with a Diffusion Prior}

\begin{document}

\twocolumn[
\icmltitle{ArrayDPS: Unsupervised Blind Speech Separation with a Diffusion Prior}




\begin{icmlauthorlist}
\icmlauthor{Zhongweiyang Xu}{UIUC}
\icmlauthor{Xulin Fan}{UIUC}
\icmlauthor{Zhong-Qiu Wang}{STU}
\icmlauthor{Xilin Jiang}{UC}
\icmlauthor{Romit Roy Choudhury}{UIUC}

\end{icmlauthorlist}

\icmlaffiliation{UIUC}{Department of Electrical and Computer Engineering, University of Illinois Urbana-Champaign, Champaign, USA}
\icmlaffiliation{UC}{Columbia University, NYC, USA}
\icmlaffiliation{STU}{Department of Computer Science and Engineering, Southern University of Science and Technology, Shenzhen, China}

\icmlcorrespondingauthor{Romit Roy Choudhury}{croy@illinois.edu}
\icmlcorrespondingauthor{Zhongweiyang Xu}{zx21@illinois.edu}

\icmlkeywords{Machine Learning, ICML}

\vskip 0.3in
]


% This command creates the footnote in the first column

\printAffiliationsAndNotice{}  % leave blank if no need to mention equal contribution

\begin{abstract}
Blind Speech Separation (BSS) aims to separate multiple speech sources from audio mixtures recorded by a microphone array. 
The problem is challenging because it is a blind inverse problem, i.e., the microphone array geometry, the room impulse response (RIR), and the speech sources, are all unknown.
We propose \textbf{\method} to solve the BSS problem in an unsupervised, array-agnostic, and generative manner. 
The core idea builds on diffusion posterior sampling (DPS), but unlike DPS where the likelihood is tractable, {\method} must approximate the likelihood by formulating a separate optimization problem.
The solution to the optimization approximates room acoustics and the relative transfer functions between microphones.
These approximations, along with the diffusion priors, iterate through the {\method} sampling process and ultimately yield separated voice sources.
We only need a simple single-speaker speech diffusion model as a prior, along with the mixtures recorded at the microphones; no microphone array information is necessary. 
Evaluation results show that {\method} outperforms all baseline unsupervised methods while being comparable to supervised methods in terms of SDR. 
Audio demos and codes are provided at:
% \hl{A sound} demo is provided at:
\href{https://arraydps.github.io/ArrayDPSDemo/}{https://arraydps.github.io/ArrayDPSDemo/} and \href{https://github.com/ArrayDPS/ArrayDPS}{https://github.com/ArrayDPS/ArrayDPS}.
% Code base in~\href{https://github.com/ArrayDPS/ArrayDPS}{https://github.com/ArrayDPS/ArrayDPS}.
\end{abstract}

% \begin{abstract}
% Blind Speech Separation (BSS) aims to separate multiple speech sources from mixtures recorded by any microphone array. 
% We propose \textbf{\method} to solve the BSS problem in an unsupervised, array agnostic, and generative manner. 
% {\method} uses a speech diffusion model as a prior and then samples from the conditional distribution of speech sources given the multi-channel mixture. 
% This is a challenging task because it is a blind inverse problem; both relative room impulse responses (relative RIR) and sources are unknown, so the likelihood is intractable. 
% {\method} solves this problem by estimating the relative RIRs during diffusion posterior sampling (DPS), initialized by independent vector analysis (IVA). 
% After the {\method} sampling process, we can separate all the sources, along with all the relative RIRs (inferring room acoustics and source locations), without knowing any of the microphone array information. 
% Also, this method only needs a simple single-speaker speech diffusion model as a prior for any microphone array. We show that our method outperforms any other unsupervised methods, while performs similarly with array-fixed supervised method. Demos can be checked here:\href{https://arraydps.github.io/ArrayDPSDemo/}{https://arraydps.github.io/ArrayDPSDemo/}.
% % We show elaborate experiments and results to validate our method.
% \end{abstract}
\vspace{-20pt}
\section{Introduction}
\label{intro}
% introduce the problem of blind source separation
The cocktail party problem is a classic challenge in audio signal processing and machine learning~\cite{cocktail2, cocktail1}. It arises when multiple speakers talk simultaneously in the same room, and several microphones capture their voices. Each microphone records a mixture of all the speakers' voices. The objective is to separate these mixtures to recover the individual voice sources. 
In recent years, supervised learning based methods have shown remarkable potential to solve the cocktail party problem~\cite{wang_supervised_2018}. 
However, these methods are usually trained with supervision from speech datasets that were synthesized by using acoustic simulators.
Such synthetic supervision inherits several problems: \\
% These methods are usually trained supervised, using synthesized datasets~\cite{wang_supervised_2018}. Because the dataset is synthesized by some acoustic simulators, deep learning models can use clean speech for supervision. However, these methods have several problems: 
{\textbf{(1) Generalizability:}} 
The simulated dataset does not match real-world acoustic environments, causing model generalization issues~\cite{generalization1, generalization2}.\\
% Unsupervised methods like UNSSOR~\cite{unssor} generalize better, but do not exploit the source prior information (e.g., clear harmonic features in speech). \\
{\textbf{(2) Deterministic:}} These models are trained to be deterministic, and hence they output one fixed separation solution for a given mixture. 
This can give blurred results when the solution is not unique~\cite{ssdgp}, i.e., the probability distribution of the sources, given the mixture, is a multi-modal distribution.\\
% (lots of separation results are plausible to form the mixture).
{\textbf{(3) Fixed Array Geometry}:} These models usually assume the geometry of microphone arrays is fixed and known, preventing flexibility to unseen arrays~\cite{arrayagnostic3}.

To address all the problems above, we propose {\method}, a generative, unsupervised, and array-agnostic algorithm for speech separation, which fully exploits speech prior information.
Building on the diffusion posterior sampling (DPS) technique~\cite{DPS, pigdm}, {\method} treats speech separation as an inverse problem.
Briefly, our goal is to recover speech sources $s$ from multi-microphone mixture measurements $x = A(s) + \epsilon$, where $A(\cdot)$ denotes the source mixing process in reverberant conditions.
DPS samples from $p(s|x)$ by using a pre-trained diffusion prior $p(s)$, and a tractable likelihood model $p(x|s)$.
While we use a pre-trained speech source diffusion prior as well, unfortunately, the likelihood is intractable in our case.
This is because the distortion function $A(\cdot)$ depends not only on the unknown array geometry but also on the unknown RIRs over which the speech arrives at each microphone.
Without any knowledge of the array geometry, the RIRs, and the speech sources, this is referred to as a ``blind'' separation problem.

% To solve this problem, our core intuition is to estimatetake advantage of the speech prior $p(s)$ modeled by a diffusion model.
% At every diffusion posterior sampling step where the likelihood is needed, we use current best source separation estimates $\hat{s}$ to estimate all the filters, which allows likelihood calculation.
% With any source estimate $\hat{s}$, and the unknown parameters in $A$, we model the measurements $\hat{x}$ as $A(\hat{s})$.
% Minimizing $\hat{x}=A(\hat{s})$ against the true measurements $x$ yields an estimate of $\hat{A}=\underset{A}{\arg\max}\; p(x|\hat{s}, A)$.
% To enable using a speech diffusion prior to solve our blind separation problem, we propose to estimate the distortion function $A$ at each diffusion posterior sampling step to enable tractable likelihood.
To solve the problems mentioned above, at each diffusion sampling step, with the current source estimate $\hat{s}$, we estimate $A$ by: $\hat{A}=\underset{A}{\arg\max}\; p(x|\hat{s}, A)$. Then we use $p(x|\hat{s}, \hat{A})$ as a tractable approximation for the intractable likelihood $p(x|\hat{s})$. Lastly, similar to DPS, we can use the prior score and our approximated likelihood score to get the posterior score, which allows posterior sampling for separation.

We borrow ideas from Forward Convolutive Prediction ~\cite{fcp, fcp2, unssor} for the estimation of $A$ mentioned above, and we use Independent Vector Analysis (IVA)~\cite{iva1, iva2} for initialization to improve sampling stability.
% use \hl{independent vector analysis}~\cite{iva1, iva2} for initialization to improve sampling stability.
% from \hl{IVA}~\cite{iva1, iva2} to improve sampling stability.
% The $\hat{s}$ and $\hat{A}$ estimates make the likelihood tractable, and samples from the likelihood are fed into the posterior score at each time step, enabling samples from the posterior.
% As the posterior sampling , ultimately offering samples of the individual speech sources. 

Extensive evaluation shows that {\method} can achieve similar performance against recent supervised methods evaluated on ad-hoc microphone arrays, and performs the best among all unsupervised blind speech separation algorithms. The gains arise from using stronger speech priors from the diffusion models, as opposed to other unsupervised methods~\cite{unssor,spatialcluster1,iva1} that do not fully exploit speech source priors. As a result, they often suffer from problems like frequency permutation or spatial aliasing, which prevents correct separation.
In contrast, {\method} leverages the diffusion prior which automatically bypasses those problems.

Our main contributions are summarized as follows:\\
\textbf{Unsupervised:} Only a clean-speech pre-trained diffusion model is needed, mitigating the generalization issues that affect supervised methods.\\
\textbf{Array-Agnostic:} 
% Modeling the problem with virtual sources (and their relative transfer functions to each microphone) allows generalization to any microphone array geometry.\\
% The method does not rely on array patterns and allows generalization on any microphone array geometries.\\
{\method} does not rely on array patterns and can generalize on any microphone array geometries.\\
\textbf{Generative:} The method samples from the posterior, allowing multiple plausible separation results while fully exploiting the speech source prior.\\
\textbf{DPS for Multi-channel:}  \method~is the first method to solve the multi-channel array inverse problem with DPS; we believe that it can enable many other applications beyond speech separation in multi-channel array signal processing.

\vspace{-6pt}
\section{Problem Formulation}\label{sec:problem}
\vspace{-1pt}
In a reverberant environment, assume a $C$-channel microphone array is recording $K$ concurrent speakers. 
Let us denote the $K$ clean speech sources as $\Tilde{s}_1(t), \Tilde{s}_2(t), ..., \Tilde{s}_K(t) \in \mathbb{R}$, where $t\in\{0,1,2,...,T-1\}$ is the sample index of the waveform.
These clean speech sources get filtered by the room impulse responses (RIR) and arrive at the $C$ microphones, forming the reverberant speech source images $s_{k,c}(t)\in \mathbb{R}$, where $k\in\{1,2,...,K\}$ indexes the sources and $c\in\{1,2,...,C\}$ indexes the microphones. Thus, each microphone captures a mixture of $K$ source images and some measurement noise, and we denote this mixture as $x_c(t)\in \mathbb{R}$.
The speaker-to-microphone filtering and mixing process can be modeled as:
% Finally, the reverberant speech images at each microphone are mixed into mixtures $x_c(t)\in \mathbb{R}$, with some measurement noise. The filtering and mixing process follows:
\vspace{-2pt}
{\small
\setlength{\abovedisplayskip}{8pt} % Reduce space above the equation
\setlength{\belowdisplayskip}{4pt} % Reduce space below the equation
\begin{align}
    s_{k,c}(t) &= h_{k,c}(t) \ast \Tilde{s}_k(t) \label{eq: filter}\\\vspace{-3pt}
    x_c(t) = \sum_{k=1}^{K} {s}_{k,c}(t) &+ n_c(t), ~~~~~~~ n_c(t)\sim\mathcal{N}(0, \sigma_n^2I)\vspace{-2pt}
     \label{eq: waveform_mix}
\end{align}
}
where $h_{k,c}(t)$ is the RIR from the $k^{th}$ source location to the $c^{th}$ microphone, and $\ast$ is the convolution operation. $n_c(t)$ is the $c^{th}$ microphone's white noise. 
Fig.~\ref{fig:signal_model}(a) illustrates the signal model for the case of $C=2$ and $K=2$.
% The signal model is shown visually in Fig.~\ref{fig:signal_model} (a). 
\textbf{Our goal is to extract all $K$ reverberant speech source images at the reference channel} ($c=1$), \textbf{i.e., extract $s_{1:K,1}$ from all the mixtures} $x_{1:C}$, \textbf{without assuming any microphone array geometry, any source location, or any supervision}\footnote{For convenience, we use $s_{k_1:k_2, c_1:c_2}$ to denote $\{{s}_{k,c}(t) | c_1\leq c \leq c_2, k_1\leq k \leq k_2\}$, and same applies to all signals like $x_{1:C}$.}.

\begin{figure}[t]
% [H]
  \centering
  \includegraphics[width=1\linewidth]{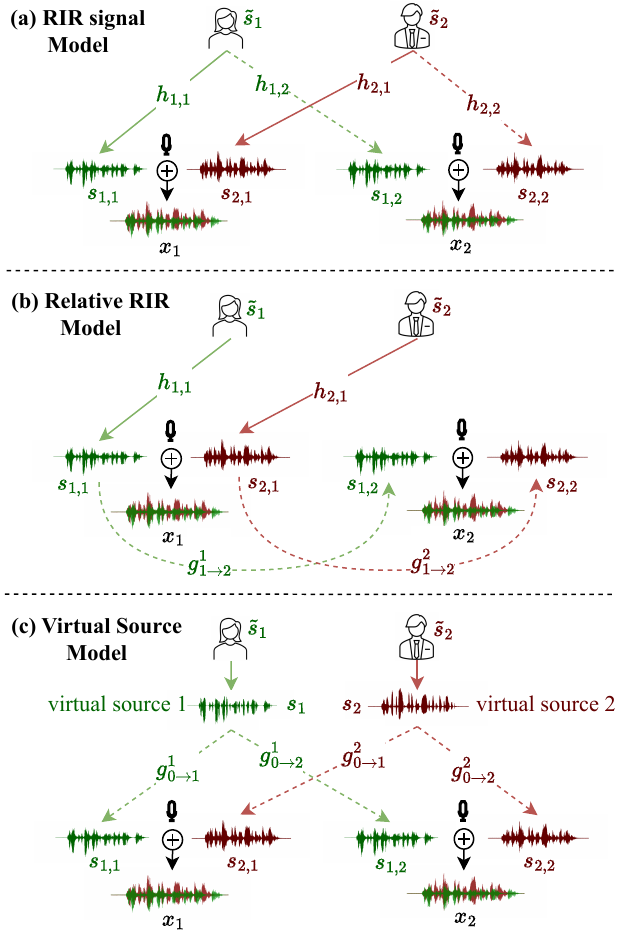}
\vspace{-20pt}
  \caption{(a) Signal mixing in the real world; (b) Relative RIR model; and (c) Relative RIR from virtual sources to real channels. Measurement noise is ignored in the figures.}
  \label{fig:signal_model}
  \vspace{-2pt}
\end{figure}
\begin{figure*}[ht]
% [H]
\vspace{-10pt}
  \centering
  \hspace*{-17pt}
  \includegraphics[width=1.03\linewidth]{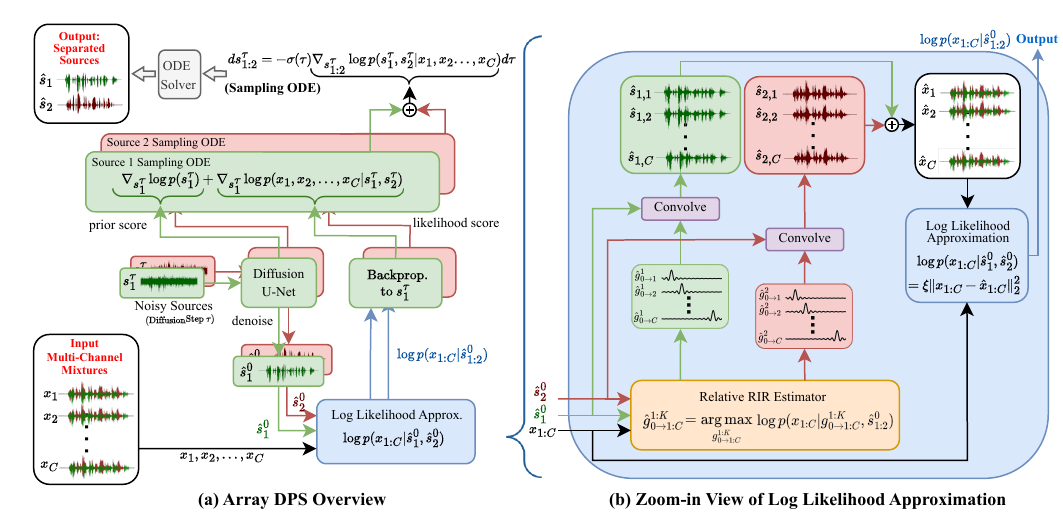}
\vspace{-20pt}
  \caption{An overview of {\method} for $K=2$ sources. Left figure (a) shows the pipeline for separation with diffusion posterior sampling, which uses the likelihood approximation pipeline shown in the right figure (b).}
  \label{fig:overview}
  \vspace{-10pt}
\end{figure*}

We approach this problem through a generative model, i.e., sampling from $p(s_{1:K,1} | x_{1:C})$.
Before formulating that, we first discuss the relative room impulse response (relative RIR) model often useful for microphone arrays (shown in Fig. \ref{fig:signal_model}(b)). 
% For array signal processing, the commonly used signal model is the relative room impulse response (relative RIR) model, as shown in Fig.~\ref{fig:signal_model}(b). 
The relative RIR models the linear relationship between any given channel and a designated reference channel (conventionally denoted as $c_1=1$).
% The relative RIR models the linear relationship between a given channel and the reference channel (channel 1). 
Note that from Eq.~\ref{eq: filter}, there exists a linear filter that can filter $s_{k, c_1}(t)$ to obtain $s_{k, c_2}(t)$ if $h_{k, c_1}(t)$ is invertible:
{\small
\setlength{\abovedisplayskip}{4pt} % Reduce space above the equation
\setlength{\belowdisplayskip}{4pt} % Reduce space below the equation
\begin{align}
    % s_{c_2, k}(t) &= h_{c_2,k}(t) \ast h^{-1}_{c_1,k}(t) \ast s_{k, c_1}(t) \\
    g^k_{c_1\rightarrow c_2}(t) &= h_{k, c_2}(t) \ast h^{-1}_{k, c_1}(t)\label{eq:relative_rir}\\
    s_{k, c_2}(t) &= g^k_{c_1\rightarrow c_2}(t) \ast s_{k, c_1}(t) \label{eq:relative_filter}
\end{align}
}
Eq.~\ref{eq:relative_rir} is called the relative RIR between channel $c_1$ and $c_2$ for speaker $k$, and Eq.~\ref{eq:relative_filter} is the convolution filtering process. 
% In applications, $c_1$ is often set to the reference channel $1$, as in Fig.~\ref{fig:signal_model} (b). 
The relative RIR model is useful because the reference channel can serve as an anchor, and all the other channel measurements can be modeled relative to this anchor.
% The relative RIR model is useful because it anchors all the processing in the reference channel, and then all other channels measurements are used  by relating to the reference channel using the relative RIRs.

\textit{Can we take one more step and avoid designating a microphone as a reference channel?}
% However, sometimes processing in reference channel might not be ideal. 
One idea is to imagine a virtual source $s_k$ for each speaker and apply a relative RIR filter to map it to a real channel. 
% pretend these are at a reference channel $c_0$, 
% One idea is to imagine a virtual microphone, i.e., a virtual channel, that can be filtered to any real channel by some relative RIR.
Fig.\ref{fig:signal_model}(c) shows an example where 
relative RIR $g^1_{0\rightarrow 1}(t)$ maps the virtual source $s_1(t)$ to the source image $s_{1,1}(t)$ at channel $1$.
% is assumed to be at channel $c_0$, and the relative RIR $g^1_{0\rightarrow 1}(t)$ maps it to channel $c_1$ to model the source image $s_{1,1}(t)$.
The same occurs for all $\langle \text{virtual source,  real channel} \rangle$ pairs:
{\small
\setlength{\abovedisplayskip}{5pt} % Reduce space above the equation
\setlength{\belowdisplayskip}{5pt} % Reduce space below the equation
\begin{equation}\label{eq:virtual_signal_model}
    s_{k,c}(t) = g^k_{0\rightarrow c}(t) \ast s_k(t), ~~~~ c\in\{1, 2, ..., C\}
\end{equation}
}
\textit{What is the advantage of modeling such a virtual source?}
The key reason is that virtual sources (along with their relative RIRs) offer the flexibility to model many system configurations.
For instance, a virtual source could be the anechoic speaker's voice signal, or it could be the measurement at a real channel $c_1$, or it could even be the signal at an imaginary microphone $c_0$ placed near the speakers.
Whatever the virtual signal represents, the corresponding relative RIR filters can adapt to match the measurements $x_{1:C}$ at the real microphones.
This flexibility allows (1) treating the reference channel $c_1$ the same as all other channels, and (2) more freedom in the optimization, allowing better performance in some algorithms~\cite{unssor}.

% Observe that the virtual source can be any signal that can be mapped to the source image of any real channel by a linear filter. For example, for source $k$, it can be the original anechoic source signal $\Tilde{s}_k(t)$, any real channel source image $s_{k,c}(t)$, or a signal that's measured by a imaginary microphone one meter away from the speaker. This flexibility allows 1) treating reference channel 1 the same other all other channels, 2) more freedom in objective optimization, allowing better performance in some algorithms~\cite{unssor}.
\vspace{-8pt}
\section{{\method}}\label{sec:method}
\vspace{-3pt}
% This section explains ArrayDPS, starting with some background and overview, followed by the algorithmic modules.

% This section explains our ArrayDPS method, starting with some background and overview, followed by our relative RIR estimation module in Sec.~\ref{sec:fcp}, the posterior score approximation in Sec.~\ref{sec:score_cond}, IVA for sampling initialization in Sec.~\ref{sec:iva}, and finally, the ArrayDPS algorithms in Sec.~\ref{sec:algo}.

% \vspace{-6pt}
\subsection{Brief Background and Overview}
\label{sec:overview}

\textbf{Background:} Our problem is formulated as $x_{1:C} = A(s_{1:K}) + n_{1:C}$, where $A$ can be understood as the filter and sum process in Fig.~\ref{fig:signal_model}(c).
% Our desired output is the separated speech from the posterior, $p(s_{1:K}|x_{1:C})$ as shown in the top of Fig.~\ref{fig:overview} (a) for $K=2$.
Separating the sources can be viewed as sampling from the posterior $p(s_{1:K}|x_{1:C})$.
For this, we use Diffusion Posterior Sampling (DPS)~\cite{DPS, pigdm, score} which samples from the posterior by solving a score-based probabilistic flow ordinary differential equation (ODE):
{\small
\setlength{\abovedisplayskip}{5pt} % Reduce space above the equation
\setlength{\belowdisplayskip}{5pt} % Reduce space below the equation
\begin{equation}\label{eq:ode_method_overview}
    ds_{1:K}^\tau=-\sigma(\tau)\nabla_{s_{1:K}^\tau} \log p(s^\tau_{1:K}|x_{1:C})d\tau
\end{equation}
}
where $s_k^\tau$ is the noisy version of $s_k$ at diffusion time~ $\tau$. 
The ODE is adapted from EDM~\cite{edm}, assuming a linear noise schedule ($\sigma(\tau)=\tau$). 
Score-based diffusion model~\cite{score, edm} suggests that sampling from $p(s_{1:K}|x_{1:C})$ is basically first initializing the sources $s_{1:K}^{\tau_\text{max}}\sim \mathcal{N}(0, \sigma^2(\tau_{\text{max}})I)$, and then getting the separated virtual sources $s_{1:K}$ by solving the ODE in Eq.~\ref{eq:ode_method_overview} until $\tau=\tau_{\text{min}}\simeq0$. 
We elaborate on the details in Appendix~\ref{app:diff}.
% $\hat{s}^0_k=s_k^{\tau}+\sigma^2(\tau)\nabla_{s_k^\tau}p(s_k^\tau)$
% Details of score-based diffusion and DPS can be checked in Appendix~\ref{app:diff}.

To acquire the posterior score $\nabla_{s_{1:K}^\tau} \log p(s^\tau_{1:K}|x_{1:C})$ in the ODE above, DPS decomposes the posterior score into a prior score and a likelihood score, where the former is estimated from a pre-trained diffusion model, and the latter is approximated based on a tractable likelihood model (also explained in Appendix~\ref{app:diff}). 
However, in our ``blind'' source separation problem, the likelihood is intractable, and solving this is at the crux of this paper.
% Thus, we approximate the intractable likelihood (in the blue module at the bottom of Fig.~\ref{fig:overview} (a) which is expanded in Fig.~\ref{fig:overview} (b)). 

\textbf{Overview:} 
Fig.~\ref{fig:overview}(a) shows the {\method} model.
The multi-channel audio mixtures $x_{1:C}$ are inputs (shown at the bottom left of the figure) and the outputs are separated virtual sources from the ODE (shown at the top of the figure).
The whole {\method} pipeline essentially develops the components needed for the ODE to reverse the diffusion process.
% In Fig.~\ref{fig:overview} (a), ArrayDPS takes the multi-channel audio mixtures $x_{1:C}$ as inputs (shown at the bottom left of the figure); and outputs the separated virtual sources from the posterior, $p(s_{1:K}|x_{1:C})$ (shown in the top of Figure).
% At a high level, the whole {\method} pipeline is solving the ODE shown in Eq.~\ref{eq:ode_method_overview} to reverse the diffusion process.
At each diffusion time $\tau$, {\method} can be summarized in five key steps (following Fig.~\ref{fig:overview}(a)):\\
\textbf{(Step 1)}:  The final posterior score $\nabla_{s_{1;K}^\tau} \log p(s^\tau_{1:K}|x_{1:C})$ in the ODE is constructed by adding the prior score $\nabla_{s_{k}^\tau}\log~p(s_{k}^\tau)$ and likelihood score $\log \nabla_{s^\tau_{k}}p(x_{1:C}|s^\tau_{1:K})$ for {\em each} source.\\
\textbf{(Step 2)}: 
The prior score is estimated using Tweedie's Formula: $\nabla_{s_k^\tau}p(s_k^\tau)=(\hat{s}^0_k-s_k^{\tau}) /\sigma^2(\tau)$. 
The $\hat{s}^0_k$ in the formula, which is the estimate of clean source, is the output of a diffusion denoising U-Net with noisy source $s_k^{\tau}$ as input.\\
% \textbf{(Step 2)}: Noisy sources $s_1^{\tau}, s_2^{\tau}$ are denoised by diffusion denoising U-Net to yield $\hat{s}^0_1, \hat{s}^0_2$. These serve as estimates of clean sources. Then the prior score can be estimated using Tweedie's Formula: $\nabla_{s_k^\tau}p(s_k^\tau)=(\hat{s}^0_k-s_k^{\tau}) /\sigma^2(\tau)$\\
\textbf{(Step 3)}:
The likelihood $p(x_{1:C}|\hat{s}_{1:K})$ is intractable because the distortion function $A(\cdot)$ (filtering+mixing) includes the unknown relative RIRs. 
However, given the relative RIRs, the log likelihood $\log~p(x_{1:C} | \hat{s}^0_{1:K}, g^{1:K}_{0\rightarrow {1:C}})$ becomes tractable. 
This motivates our {\em log likelihood approximation} model expanded in Fig.~\ref{fig:overview}(b).

% Thus, we propose to estimate the relative RIRs and then use them to calculate our tractable log likelihood as shown in Fig.~\ref{fig:overview} (b).
Given the measured mixture $x_{1:C}$ and the denoised estimates $\hat{s}^0_{1:K}$ from time $\tau$, the yellow module in Fig.~\ref{fig:overview}(b) estimates the relative RIRs $\hat{g}^{1:K}_{0\rightarrow {1:C}}$ in a maximum likelihood manner (detailed next in Sec.~\ref{sec:fcp}).
% First, following the yellow module, the virtual source estimates $\hat{s}^0_{1:2}$ and the mixture $x_{1:C}$ are used to estimate the relative RIRs $\hat{g}^{1:K}_{0\rightarrow {1:C}}$ in a maximum likelihood manner (discussed in next Sec.~\ref{sec:fcp}). 
Recall these relative RIRs are transfer functions from the virtual sources to the real microphone channels, hence appropriately convolving $\hat{g}^{1:K}_{0\rightarrow {1:C}}$ with the virtual source estimates $\hat{s}_{1:K}$ gives us the {\em actual} source image estimates $\hat{s}_{1:K, 1:C}$ at all the microphone channels.
% These relative RIRs can be appropriately convolved with the virtual source estimates $\hat{s}_{1:K}$ to obtain the {\em actual} source estimates $\hat{s}_{1:K, 1:C}$ at all the microphone channels.
Adding up these actual source image estimates yields the estimated mixtures $\hat{x}_{1:C}$. 
% The estimated $\hat{g}^{1:K}_{0\rightarrow {1:C}}$ are then used to filter the virtual source estimates $\hat{s}_{1:K}$ to get source estimates $\hat{s}_{1:K, 1:C}$ at each of the channels.
% These can then be added to obtain the estimated mixtures $\hat{x}_{1:C}$. 
Finally, based on the Gaussian noise model $n_c(t)$ in Eq.~\ref{eq: waveform_mix}, the log likelihood can be directly calculated as $\|x_{1:C}-\hat{x}_{1:C}\|_2^2$, weighted by a signal energy-based scaler $\xi$.\\
% Our log likelihood approximation module (in blue) uses the denoised sources $\hat{s}^0_{1:K}$ and the recorded mixture $x_{1:C}$ to approximate $\log~p(x_{1:C}|\hat{s}^0_{1},\hat{s}^0_{2})$. Details will be discussed in the next paragraph following Fig.~\ref{fig:overview} (b)\\
 \textbf{(Step 4)}: Back in Fig.~\ref{fig:overview}(a), the approximated log likelihood is back-propagated to $s_k^\tau$ to compute the likelihood score. \\
\textbf{(Step 5)}: This likelihood score and the prior score (from step 2) are combined to obtain the posterior score at $\tau$.

% \textbf{Log Likelihood Approximation Overview:}
% To approximate log likelihood $\log p(x_{1:C}|\hat{s}^0_{1:K})$, we first notice that it is not directly tractable, because we do not know the relative RIR filters $g^{1:K}_{0\rightarrow {1:C}}$. (Note that in DPS this likelihood has to be tractable.) According to the virtual source model as in Fig.~\ref{fig:signal_model} (c), $\log p(x_{1:C} | \hat{s}^0_{1:K}, g^{1:K}_{0\rightarrow {1:C}})$ is tractable because the relative RIR is given. Thus, we propose to estimate a relative RIR and then use the estimated one to calculate the log likelihood. Our tractable log likelihood approximation is then shown in Fig.~\ref{fig:overview} (b) including three steps:\\
% \textbf{(Step 1)}
% As shown in the yellow module, We first estimate the relative RIRs in a maximum likelihood manner. The relative RIR estimator is able to estimate $g^{1:K}_{0\rightarrow {1:C}})$ from the denoised sources $\hat{s}^0_{1:K}$ and all channel mixtures $x_{1:C}$, with differentiable analytical solution. Details of the estimator will be discussed in the next subsection (Sec.~\ref{sec:fcp}).\\
% \textbf{(Step 2)}
% Then the estimated relative RIRs are used to filter the estimated sources to get all-channel estimated source images $\hat{s}_{k,c}$, which can then be used to mix the estimated mixtures $\hat{x}_{1:C}$.\\
% \textbf{(Step 3)}
% Lastly, based on the Gaussian likelihood noise model as in Eq.~\ref{eq: waveform_mix}, the log likelihood can be directly calculated as $\|x_{1:C}-\hat{x}_{1:C}\|_2^2$ multiplied by some scaler $\xi$.

Note that Fig.~\ref{fig:overview}(a) is separating the {\em virtual} sources $s_{1:K}$. 
Since we want to estimate the source images at the reference channel, namely $s_{1:K,1}$, we convolve with the estimated relative RIRs $\hat{g}^k_{0\rightarrow 1}$ to obtain the separated source images $\hat{s}_{k,1}$.

\vspace{-4pt}
\subsection{Relative RIR Estimation}\label{sec:fcp}
Since log likelihood $\log~p(x_{1:C} | \hat{s}^0_{1:K}, g^{1:K}_{0\rightarrow {1:C}})$ is tractable, we propose to estimate the $g^{1:K}_{0\rightarrow {1:C}}$ in a maximum likelihood manner using mixtures $x_{1:C}$ and denoised sources $\hat{s}^0_{1:K}$:
{\small
\setlength{\abovedisplayskip}{5pt} % Reduce space above the equation
\setlength{\belowdisplayskip}{1pt} % Reduce space below the equation
\begin{align}
        \hat{g}^{1:K}_{0\rightarrow {1:C}}~~=~~\underset{g^{1:K}_{0\rightarrow {1:C}}}{\arg\max}\;~~ \log\;~p(x_{1:C} | g^{1:K}_{0\rightarrow {1:C}}, \hat{s}^0_{1:K})\label{eq:relativerir_ml}
\end{align}}
However, previous work has shown that it is better to estimate the filters in spectral domain~\cite{eras, consolidated} instead of in the time domain shown above. 
Hence, we transform to the Short-Time-Fourier-Transform (STFT) domain.
In STFT domain, Eq.~\ref{eq:virtual_signal_model} becomes:
{\small
\setlength{\abovedisplayskip}{5pt} % Reduce space above the equation
\setlength{\belowdisplayskip}{5pt} % Reduce space below the equation
\begin{align} 
  S_{k,c}(l, f) &= G^k_{0\rightarrow c}(f) \cdot S_{k}(l, f)\label{eq:rtf}
\end{align}
}
where $S_{k,c}(l, f)$ and $S_{k}(l, f)$ are the STFT of $s_{k,c}(t)$ and $s_k(t)$, respectively. 
$G^k_{0\rightarrow c}(f)$ is the discrete Fourier transform of $g^k_{0\rightarrow c}(t)$, where the FFT size is the same as the STFT. 

\textbf{Relaxing narrowband assumption:} The above equation is unrealistic for speech signals because it makes a narrowband approximation; it assumes that the length of the relative RIR filter is shorter than the FFT size. 
In real environments, the relative RIR filter can be in hundreds of milliseconds, while the FFT size is usually only tens of milliseconds. 
To relax this assumption, we model the relative RIR over multiple time-frames as $G^k_{0\rightarrow c}(l, f)$.
This multi-frame filter in the STFT domain is then {\em convolved} with the virtual source as follows:
% {\small
% \setlength{\abovedisplayskip}{3pt} % Reduce space above the equation
% \setlength{\belowdisplayskip}{3pt} % Reduce space below the equation
% \begin{align} 
%   S_{k,c}(l, f) &= G^k_{0\rightarrow c}(l, f) *_l S_{k}(l, f)\label{eq:stft_relative}\\
%    G^k_{0\rightarrow c}(l, f) *_l S_{k}(l, f) &= \sum_{j=-F}^{P} G^k_{0\rightarrow c}(j, f) S_{k}(l-j, f)\label{eq:stft_relative2}
% \end{align}
% }
{\small
\setlength{\abovedisplayskip}{6pt} % Reduce space above the equation
\setlength{\belowdisplayskip}{6pt} % Reduce space below the equation
\begin{align} 
   S_{k,c}(l, f) &= G^k_{0\rightarrow c}(l, f) *_l S_{k}(l, f)\label{eq:stft_relative}
\end{align}
}
where $l,f$ are frame and frequency index, and $F, P$ are the number of future and past frames of the relative RIR $G^k_{0\rightarrow c}(l, f)$, all in STFT domain.
Observe that $*_l$ in~Eq.\ref{eq:stft_relative} denotes {\em convolution} in the frame dimension:
{\small
\setlength{\abovedisplayskip}{3pt} % Reduce space above the equation
\setlength{\belowdisplayskip}{3pt} % Reduce space below the equation
\begin{align} 
   G^k_{0\rightarrow c}(l, f) *_l S_{k}(l, f) &= \sum_{j=-F}^{P} G^k_{0\rightarrow c}(j, f) S_{k}(l-j, f)\label{eq:stft_relative2}
\end{align}
}
Thus, the complete signal model --- that connects the virtual sources to the actual mixtures at the microphones --- can be modeled as:
{\small
\setlength{\abovedisplayskip}{3pt} % Reduce space above the equation
\setlength{\belowdisplayskip}{3pt} % Reduce space below the equation
\begin{align}\label{eq:spec_model}
    X_c(l, f) =\sum\limits_{k=1}^{K} G^k_{0\rightarrow c}(l, f)*_l S_{k}(l, f) + N_c(l, f)
\end{align}
}
\textbf{ML Estimation:} Under this setting, we intend to estimate relative RIRs $G^k_{0\rightarrow c}(l,f)$ in a maximum likelihood manner.
We found that the Forward Convolutive Prediction (FCP) method~\cite{fcp, fcp2, unssor} is able to achieve this. 
Basically, FCP estimates relative RIRs $G^k_{0\rightarrow c}(l, f)$ from the multi-channel mixtures $X_c(l, f)$ and the virtual source estimate $\hat{S}_{k}(l, f)$ by solving a linear problem:
{\fontsize{8.5pt}{11pt}\selectfont
\setlength{\abovedisplayskip}{5pt} % Reduce space above the equation
\setlength{\belowdisplayskip}{5pt} % Reduce space below the equation
\begin{align}
    % \hat{\lambda}^k_c(l,f) &= \text{some expression defining lambda} \\
    \hat{G}^k_{0\rightarrow c}(l, f) &= \text{FCP} (X_c(l, f), \hat{S}_{k}(l, f))\label{eq:fcp_func}\\
    = \underset{G^k_{0\rightarrow c}(l, f)}{\arg\min} &\sum_{l,f} \frac{1}{\hat{\lambda}^k_c(l,f)} \left|X_c(l, f) - G^k_{0\rightarrow c}(l, f)*_l\hat{S}_{k}(l, f)\right|^2\label{eq:fcp_obj}\\
    % &\hspace{4cm}+\epsilon\left\|G^k_{1,c}\right\|_2^2\\
    % &\hat{\lambda}^k_c(l,f) = 
    \hat{\lambda}^k_c(l,f) = &\frac{1}{C} \sum_{c=1}^C |X_{c}(l, f)|^2+ \epsilon\cdot \max\limits_{l,f}\frac{1}{C} \sum_{c=1}^C |X_{c}(l,f)|^2\label{eq:fcp_lambda}
\end{align}
}
As shown above, FCP is solving a weighted least squares problem, so it has an analytical solution as in~\cite{unssor}. 
The weight $\hat{\lambda}^k_c(l,f)$ aims to prevent overfitting to high-energy STFT bins, and $\epsilon$ is a parameter to adjust the weight. 
Note that two important parameters for FCP are $F$ and $P$ from Eq.~\ref{eq:stft_relative2}, which control the past and future filter lengths of the relative RIR. 
We prove FCP is equivalent to the maximum likelihood relative RIR estimator in Appendix~\ref{app:analysis} Theorem~\ref{thm1}. 
Thus, with the FCP estimated filters, we can obtain the estimated source images at all microphone channels:
{\small
\setlength{\abovedisplayskip}{5pt} % Reduce space above the equation
\setlength{\belowdisplayskip}{5pt} % Reduce space below the equation
\begin{equation}\label{eq:fcp}
    % \hat{S}^{\text{FCP}}_{k,c}(l, f) = \sum_{j=-F}^{P} \hat{G}^k_{1,c}(j, f) \hat{S}_{1, k}(l-j, f)
    \hat{S}_{k,c}(l, f) = \hat{G}^k_{0\rightarrow c}(l, f) *_l \hat{S}_{k}(l, f)
\end{equation}
}
\vspace{-20pt}
\subsection{Posterior Score Approximation}\label{sec:score_cond}
\vspace{-5pt}
The log likelihood approximation gives us relative RIR estimates and the corresponding source image estimates at each microphone channel.
All of these were computed based on clean source estimates denoised by the diffusion model at time $\tau$.
Using all the available estimates, {\method} now needs to compute the posterior score, and then run the ODE solver in Eq.~\ref{eq:ode_method_overview} to ultimately get samples from $p(s_{1:K}|x_{1:C})$.
We describe this process next.
% As mentioned in Sec.~\ref{sec:problem}, \method~aims to solve the separation problem by sampling from $p(s_{1:K}|x_{1:C})$, i.e., sampling the virtual  sources given the multi-channel mixtures. According to Eq.~\ref{eq:ode_method_overview}, to allow for diffusion posterior sampling, we need to run an ODE solver on:
% {\small
% \begin{equation}\label{eq:separationode}
% d{s}^\tau_{1:K} = -\sigma(\tau) \nabla_{s^\tau_{1:K}} \log p(s^\tau_{1:K} | x_{1:C}) d\tau
% \end{equation}
% }

Assume we have a pre-trained diffusion denoising model $D_\theta(s^\tau_{k}, \sigma(\tau))$ trained on single-channel speech sources; the corresponding score model $S_\theta(s^\tau_{k}, \sigma(\tau))$ approximates $\nabla_{s^\tau_{k}} \log ~p (s^\tau_{k})$. 
Based on this, we present here a novel approximation of the posterior score $\nabla_{s^\tau_{1:K}}\log p(s^\tau_{1:K} | x_{1:C})$ to enable sampling using the ODE above. 
All the derivation, proof, and analysis of our approximation is in Appendix~\ref{app:posterior_score}. 
Below is our final result, where we first denoise the noisy sources $s^\tau_{1:K}$ and then transform to STFT domain:
{\small
\setlength{\abovedisplayskip}{5pt} % Reduce space above the equation
\setlength{\belowdisplayskip}{3pt} % Reduce space below the equation
\begin{equation}
    \hat{s}^0_k = D_\theta(s^\tau_k, \sigma(\tau)),\; \quad \;\hat{S}^{0}_{k}=\text{STFT}(\hat{s}^0_k),\;~~\;k\in\{1,...,K\}
\end{equation}
}
Then we use FCP to estimate all relative RIRs:
{\small
\setlength{\abovedisplayskip}{3pt} % Reduce space above the equation
\setlength{\belowdisplayskip}{3pt} % Reduce space below the equation
\begin{align}
    \hat{G}^{k}_{0\rightarrow c} &= \text{FCP}(X_{c},\;\;\hat{S}^{0}_{k}), \; \quad \;\hat{S}_{k,c}= \hat{G}^{k}_{0\rightarrow c} *_l \hat{S}_{k}
\end{align}
}
% \hl{Need one sentence here on the backpropagation step to compute the likelihood score using the ouput of our likelihood approximation module.}
The relative RIRs enable log likelihood estimates (at the output of Fig.~\ref{fig:overview}(b)), and to get their respective likelihood scores, we back-propagate to $s^{\tau}_1, s^{\tau}_2$ to calculate the gradient.
In the final step, the likelihood scores and the prior scores are summed up to give the posterior score:
% Thus, our final approximated posterior score is the sum of the prior and likelihood scores as follows:
{\fontsize{9pt}{0pt}\selectfont
\vspace{-5pt}
\setlength{\abovedisplayskip}{5pt} % Reduce space above the equation
\setlength{\belowdisplayskip}{5pt} % Reduce space below the equation
\begin{align}
    \nabla&_{s^\tau_{1:K}} \log p(s^\tau_{1:K} | x_{1:C}) \nonumber~~\simeq~~ \sum_{k=1}^{K}S_\theta(s^\tau_{k}, \sigma(\tau))\nonumber ~~+ \\
    & \quad \quad \quad \sum_{c=1}^{C}\nabla_{s^\tau_{1:K}}\xi(\tau)\left\|x_c-\text{ISTFT}\left(\sum_{k=1}^{K}\hat{S}^{0}_{k,c}\right)\right\|_2^2\label{eq:final_score}
\end{align}
}
ISTFT above denotes Inverse-Short-Time-Fourier-Transform. Note that the functional relationship between $s^\tau_{1:K}$ and $\hat{S}^{0}_{1:K,1:C}$ is fully differentiable, as operations $D_\theta$, STFT, ISTFT and FCP are all differentiable. 
Thus, in the second term in Eq.~\ref{eq:final_score}, the likelihood score can be calculated by back-propagating to $s^\tau_{1:K}$. 
Again, detailed proof of correctness, analysis, and relation to Expectation Maximization is included in Appendix~\ref{app:posterior_score}.

\vspace{-5pt}
\subsection{IVA as Initialization}\label{sec:iva}
\vspace{-1pt}
At the early stages of solving the separation ODE, the source estimates $\hat{s}_{1:K}^{0}$ are extremely inaccurate, which means the estimated relative RIRs are inaccurate. 
To mitigate this, we propose to use an Independent Vector Analysis algorithm~\cite{iva1,iva2} as initialization, for both the sources and the relative RIRs.

IVA is a classic unsupervised algorithm for multi-channel blind speech separation, which has limited performance in reverberant conditions~\cite{iva1, iva2}. For initialization purposes, we use IVA to separate the sources at the reference channel ($c_1=1$) and then use these separated sources to further estimate the relative RIRs $\hat{G}^{k,\text{IVA}}_{1\rightarrow c}$:
{\small
\setlength{\abovedisplayskip}{1pt} % Reduce space above the equation
\setlength{\belowdisplayskip}{1pt} % Reduce space below the equation
\begin{align}
    \hat{S}^{\text{IVA}}_{1,1}, ~~... ~~ \hat{S}^{\text{IVA}}_{K,1} &~=~ \text{IVA}(X_{1:C})\\
    \hat{G}^{k,\text{IVA}}_{1\rightarrow c} ~=~ \text{FCP}(X_c, \hat{S}^{\text{IVA}}_{k,1})&\text{, } ~~k\in\{1,...,K\},~~c\in\{1,...,C\}
\end{align}
}
Then the IVA separated sources $\hat{S}^{\text{IVA}}_{1:K,1}$ are used to initialize the starting point of posterior sampling, and the estimated relative RIRs $\hat{G}^{k,\text{IVA}}_{1\rightarrow c}$ are used to guide the posterior sampling in early stages, which will be discussed next in Sec.~\ref{sec:algo}.
% \begin{equation}
% S_{1,k}(l,f)\sim \mathbb{CN}(0, \sigma^2_k)
% \end{equation}

\vspace{-5pt}
\subsection{Algorithm}\label{sec:algo}
\vspace{-1pt}
Our proposed {\method} is specified in Algorithm~\ref{alg:ubliss}, which calls the posterior score approximation function, \texttt{GET\_SCORE($\cdot$)}, presented in Algorithm~\ref{algo:score_cond}. 
We intuitively explain blocks of the algorithm below while the details can be found in Appendix~\ref{app:algorithms}.
\begin{algorithm}[tb]
\caption{\method}
\label{alg:ubliss}
\begin{algorithmic}[1]
\Require $\{N, \sigma_{i\in\{0,...,N\}}, \gamma_{i\in\{0,...,N-1\}}, S_{\text{noise}}\}$
% \Statex \textbf{Inputs}: $\{x_{1:C}, K\}$  % This line will not be numbered
\Statex \hspace{-\algorithmicindent}\textbf{Inputs}: mixtures $x_{1:C}$, number of sources $K$  % This line uses negative horizontal space to shift left
\State $X_{1:C} \gets \text{\small{STFT}}(x_{1:C})$
\State $\hat{S}^{\text{{IVA}}}_{1,1},..., \hat{S}^{\text{{IVA}}}_{K,1} \gets \text{{IVA}}(X_{1:C})$\Comment{{\textcolor{magenta}{IVA separation}}}
\State $\hat{s}^{\text{{IVA}}}_{1,1},..., \hat{s}^{\text{{IVA}}}_{K,1} \gets \text{\small{ISTFT}}(\hat{S}^{\text{{IVA}}}_{1,1}),..., \text{\small{ISTFT}}(\hat{S}^{\text{\small{IVA}}}_{K,1})$
\For{$k=1$ \textbf{to} $K$, $c=1$ \textbf{to} $C$}
    \State $\hat{G}^{k,\text{{IVA}}}_{1\rightarrow c} \gets \text{\small{FCP}}(X_c, \hat{S}^{\text{{IVA}}}_{k,1})$ \Comment{\small{\textcolor{magenta}{IVA Init. relative RIR}}}
\EndFor
\For{$k=1$ \textbf{to} $K$}
    \State Initialize $s^{\tau_0}_{k} \sim \mathcal{N}(\hat{s}^{\text{IVA}}_{k,1}, \sigma_0^2I)$\Comment{\textcolor{magenta}{\small{Diff. Init. with IVA}}}
\EndFor
\For {$i=0$ \textbf{to} $N-1$} \Comment{\textcolor{magenta}{\small{$2^{\text{nd}}$order Heun stochastic sampler}}}
\State $\hat{\sigma}_i\gets {\sigma}_i + \gamma_i \sigma_i$
\State Sample $\epsilon_i\sim\mathcal{N}({0, S_{noise}^2 I})$
\For{$k=1$ \textbf{to} $K$}
\State $\hat{s}^{\tau_i}_{k} \gets s^{\tau_i}_{k} + \sqrt{\hat{\sigma}_i^2-\sigma_i^2}\epsilon_i$ \Comment{\small{\textcolor{magenta}{add stochasticity}}}
\EndFor
\State $score \gets$ {\scriptsize \texttt{GET\_SCORE}}($\hat{s}^{\tau_i}_{1:K}$, $x_{1:C}$, $\hat{G}^{k,\text{{IVA}}}_{1\rightarrow c}$, $i$) \Comment{\small{{\textcolor{magenta}{Algo.\ref{algo:score_cond}}}}}
\State $d_i \gets -\hat{\sigma}_i\cdot score$
\State $s^{\tau_{i+1}}_{1:K}\gets \hat{s}^{\tau_{i}}_{1:K}+(\sigma_{i+1}-\hat{\sigma}_i)\cdot d_i$ \Comment{\textcolor{magenta}{$1^\text{st}$ order Euler step}}
\If{$\sigma_i\neq 0$}\Comment{\textcolor{magenta}{$2^\text{nd}$ order correction}}
\State $score' \gets$  {\scriptsize \texttt{GET\_SCORE}}($s^{\tau_{i+1}}_{1:K}$, $x_{1:C}$, $\hat{G}^{k,\text{{IVA}}}_{1\rightarrow c}$, $i+1$)
\State $d_i'\gets-\sigma_{i+1}\cdot$ $score'$
\State $s^{\tau_{i+1}}_{1:K}\gets \hat{s}^{\tau_{i}}_{1:K}+\frac{1}{2}(\sigma_{i+1}-\hat{\sigma}_i)\cdot (d_i + d_i')$
\EndIf
\EndFor
% \State \Return $s^{\tau_{N}}_{1:K}$ \Comment{{\small \textcolor{magenta}{return separated virtual sources}}}
\State $\hat{S}_{1:K} \gets \text{\small{STFT}}(s^{\tau_{N}}_{1:K})$
\For{$k=1$ \textbf{to} $K$} \Comment{{\small \textcolor{magenta}{transform to channel 1 with FCP}}}
\State $\hat{G}^{k,\text{final}}_{0\rightarrow 1} \gets \text{\small{FCP}}(X_1, \hat{S}_{k})$
\State $\hat{s}_{k,1} \gets \text{ISTFT}(\hat{G}^{k,\text{final}}_{0\rightarrow 1} * \hat{S}_{k})$
\EndFor
% \State \Return $\hat{s}_{1:K,1}$ \Comment{{\small \textcolor{magenta}{return sources in reference channel 1}}}
\State \Return $s^{\tau_{N}}_{1:K}$, $\hat{s}_{1:K,1}$ \Comment{{\small \textcolor{magenta}{separated virtual and channel 1 sources}}}
% \EndFor
\end{algorithmic}
\end{algorithm}

\vspace{-2pt}
{\bf \method~Algorithm:}
The \method~algorithm in Algorithm~\ref{alg:ubliss} needs a few sampler parameters to be pre-defined: sampling steps and schedule $N$, $\sigma_{i\in\{0,...,N-1\}}$, and $\gamma_{i\in\{0,...,N-1\}}$, $S_{\text{noise}}$, which all follow the stochastic sampler in EDM~\cite{edm}. 
The algorithm then takes the mixtures $x_{1:C}$ and number of sources $K$ as inputs, to separate source images in the reference channel. 
From lines 1-6, IVA is used to separate reference-channel source images, which are further used to estimate the relative RIRs. 
Then in lines 7-9, the IVA separated sources are used as diffusion initialization. 
Line 10-24 is a second-order Heun stochastic sampler following EDM~\cite{edm}, where the score is using our Posterior Score Approximation, \texttt{GET\_SCORE($\cdot$)}, in Algorithm~\ref{algo:score_cond}. 
Note that the sampler is sampling the virtual sources instead of reference channel source images; we discuss the reason later in Appendix~\ref{app:fcp}. 
Then, to output separated source images in the reference channel, lines 25-30 use FCP to estimate the relative RIRs that allow filtering virtual sources to reference channel source images.
\begin{algorithm}[tb]
\caption{Posterior Score Approximation} \label{algo:score_cond}
\begin{algorithmic}[1]
\Require $\{D_\theta, \sigma_{i\in\{0,...,N\}}, N_{\text{ref}}, N_{\text{fg}}, \xi_1(\tau), \xi_2(\tau), \lambda\}$
\Function{\texttt{get\_score}}{$s^{\tau_i}_{1:K}$, $x_{1:C}$, $\hat{G}^{k,\text{IVA}}_{1\rightarrow c}$, $i$}
\For{$k=1$ \textbf{to} $K$}
    \State $\hat{s}^0_k \gets D_\theta(s^{\tau_i}_{k}, \sigma_i)$\Comment{\small{\textcolor{magenta}{diffusion denoiser}}}
    \State $\hat{S}^0_k \gets \text{STFT}(\hat{s}^0_k)$
    % \State $\hat{S}_k \gets \text{\small{STFT}}(D_\theta(s^{\tau_i}_{0, k}, \sigma_i))$ \Comment{\small{diffusion denoiser}}
\EndFor
\State $c_i\gets \frac{\hat{s}^0_{1:K} - s^{\tau_i}_{1:K}}{\sigma^2_i}$ \Comment{\small{\textcolor{magenta}{get prior score}}}
\State $X_{1:C}\gets \text{\small{STFT}}(x_{1:C})$
\For{$k=1$ \textbf{to} $K$, $c=1$ \textbf{to} $C$}
\If{$i\leq N_{\text{fg}}$}
\State $\hat{G}^{k}_{0\rightarrow c}\gets \text{\small{FCP}}(X_{c}, \hat{S}^0_k)$ \Comment{\small{\textcolor{magenta}{Est. relative RIR}}}
\Else
\State $\hat{G}^{k}_{0\rightarrow c}\gets \hat{G}^{k,\text{IVA}}_{1\rightarrow c}$ \Comment\small{{\textcolor{magenta}{Use IVA init. relative RIR}}}
\EndIf
% \State $\hat{S}_{k,c}(l,f)\gets \sum_{j=-F}^{P} \hat{G}^{k}_{0,c}(j, f) \hat{S}_{k}(l-j, f)$
\State $\hat{s}_{k,c}\gets \text{\small{ISTFT}}(\hat{G}^{k}_{0\rightarrow c} *_l \hat{S}^0_{k})$ \Comment{\textcolor{magenta}{filter to all channels}}
\EndFor
\State $g_i \gets -\xi_1(\tau_i)\nabla_{s^\tau_{1:K}}\sum\limits_{c=1}^{C}\left\|x_c-\sum_{k=1}^{K}\hat{s}_{k,c}\right\|_2^2$
\If{$i\leq N_{\text{ref}}$} \Comment{\textcolor{magenta}{reference channel guidance}}
\State $r_i \gets -\xi_2(\tau_i)\nabla_{s^\tau_{1:K}}\left\|x_1-\sum_{k=1}^{K}\hat{s}^0_{k}\right\|_2^2$
\State $g_i \gets g_i + \lambda r_i $ \Comment\small{\textcolor{magenta}{update posterior score}}
\EndIf
% \State $c_i\gets c_i + g_i$
\State \Return $c_i + g_i$ \Comment{\small{\textcolor{magenta}{final posterior score}}}
\EndFunction
% \If{$i\leq N_{ref}$}
% \State $g_i \gets g_i - \lambda \xi(\tau_i)\nabla_{s^\tau_{0,1:K}}\sum\limits_{c=1}^{C}\left\|x_c-\sum_{k=1}^{K}\hat{s}_{0, k}\right\|_2^2$
% \EndIf
\end{algorithmic}
\end{algorithm}

\vspace{-1pt}
{\bf Posterior Score Approximation Algorithm:}
The posterior score approximation in Algorithm~\ref{algo:score_cond} follows the result in Sec.~\ref{sec:score_cond} and provides Algorithm~\ref{alg:ubliss} with the posterior score for diffusion sampling.
Algorithm~\ref{algo:score_cond} needs a few pre-defined parameters. 
$D_\theta$ is pre-trained diffusion denoiser. 
$N_{\text{ref}}$ is the number of early steps needed for reference channel guidance, where the sum of estimated sources are directly guided by the reference channel mixture (lines 17-19). 
We found this empirical term helps improve algorithm robustness and convergence. 
$N_{\text{fg}}$ means in the first $N_{\text{fg}}$ sampling steps, the IVA initialized relative RIRs are used to calculate the likelihood, instead of the FCP-estimated ones. 
$\xi_1(\tau), \xi_2(\tau), \lambda$ are all weighting scalars which will be discussed in detail in Appendix~\ref{app:algorithms}. 
The algorithm takes the current diffusion noisy sources $s^{\tau_i}_{1:K}$, the mixtures $x_{1:C}$, the IVA initialized relative RIRs $\hat{G}_{1\rightarrow c}^{k, \text{IVA}}$ and the diffusion step $i$ as inputs. 
% In lines 2-6, the diffusion denoiser denoises the sources and then uses them to estimate the prior score (using Tweedie's Formula).
In lines 2-6, the noisy sources are denoised, and are then used to estimate the prior score (using Tweedie's Formula).
Then in lines 9-15, the relative RIRs (estimated from IVA init. or estimated using current source est.) are used to transform estimated virtual sources to all channels. 
In line 16 the approximated likelihood score is estimated. 
Line 17-19 is the reference channel guidance mentioned. 
Then in line 21, the posterior score is the sum of the prior score and likelihood score.
\begin{table*}[h]
\vspace{-10pt}
\centering
\caption{Evaluation results for 3-channel SMS-WSJ. Methods denoted with $\dagger$ are results from UNSSOR ~\cite{unssor}, and methods denoted with $*$ mean it is impractical. Note that SMS-WSJ only contains samples with a fixed microphone array. Top results are color-coded as \colorbox{top1}{top1}, \colorbox{top2}{top2}, and \colorbox{top3}{top3}.}
\setlength\tabcolsep{4.6pt}
\renewcommand{\arraystretch}{0.7}  % Reduced line spacing
\label{tab:smswsj}
\small
\begin{tabular}{@{}cccccccccccc@{}} % Added a column for the row numbers
\toprule
 \multirow{2}{*}{Row} & \multirow{2}{*}{Methods} & \multirow{2}{*}{Unsup.} & \multirow{2}{*}{\shortstack{Array \\[-3pt] Agnostic}} & \multirow{2}{*}{Prior} & \multirow{2}{*}{\shortstack{IVA \\ Init.}} & \multirow{2}{*}{\shortstack{Ref. \\ Guide.}} & \multirow{2}{*}{\shortstack{SDR \\[-1pt] (dB)}} & \multirow{2}{*}{\shortstack{SI-SDR \\[-1pt] (dB)}} & \multirow{2}{*}{PESQ} & \multirow{2}{*}{eSTOI} & \\
 & & & & & & & & & & & \\
\midrule
\rowcolor{gray!10}
0 & Mixture  & - & - & - & - & - & 0.1 & 0.0 & 1.87 & 0.603 & \\
\midrule
1a & Spatial Cluster$^\dagger$ & \checkmark & \checkmark & - & - & - & 9.5 & 8.5 & 2.52 & 0.759 & \\
\rowcolor{gray!10}
1b & IVA$^\dagger$ & \checkmark & \checkmark & Laplace & - & - & 12.0 & 10.7 & 2.67 & 0.802 & \\
1c & IVA & \checkmark & \checkmark & Gaussian & - & - & 13.4 & 12.2 & 2.82 & 0.834 & \\
\rowcolor{gray!10}
1d & UNSSOR$^\dagger$ & \checkmark & \texttimes & - & - & - & 15.4 & 14.4 & 3.20 & 0.875 & \\
\midrule
% 2a & \multirow{5}{*}{\shortstack{UBliSS-Diff \\ (A,B,C,D,E)}} & \checkmark & \checkmark & Anechoic & \checkmark & \checkmark & 15.7$\pm$1.2 & 14.9$\pm$1.3 & 3.36$\pm$0.12 & 0.863$\pm$0.020 & \\
2a & \method-A & \checkmark & \checkmark & Anechoic & \checkmark & \checkmark & 15.8$\pm$1.2 & 15.0$\pm$1.3 & 3.38$\pm$0.12 & 0.865$\pm$0.020 & \\
\rowcolor{gray!10}
2b & \method-B & \checkmark & \checkmark & Reverb. & \checkmark & \checkmark & 15.1$\pm$1.1 & 14.3$\pm$1.2 & 3.29$\pm$0.11 & 0.850$\pm$0.019 & \\
2c & \method-C & \checkmark & \checkmark & Anechoic & \checkmark & \texttimes & 14.9$\pm$1.3 & 14.0$\pm$1.5 & 3.16$\pm$0.16 & 0.844$\pm$0.026 & \\
\rowcolor{gray!10}
2d & \method-D & \checkmark & \checkmark & Anechoic & \texttimes & \checkmark & 8.5$\pm$4.6 & 6.8$\pm$5.2 & 2.59$\pm$0.48 & 0.731$\pm$0.111 & \\
% 2e & & \checkmark & \checkmark & Anechoic & \texttimes & \texttimes & 1.3$\pm$2.0 & -0.7$\pm$2.3 & 1.79$\pm$0.18 & 0.511$\pm$0.059 & \\
2e & \method-E & \checkmark & \checkmark & Anechoic & \texttimes & \texttimes & 0.9$\pm$1.6 & -1.4$\pm$1.9 & 1.74$\pm$0.15 & 0.518$\pm$0.058 & \\
\midrule
% 3a & \multirow{2}{*}{\shortstack{UBliSS-Diff-Max \\ (A,B)}} & \checkmark & \checkmark & Anechoic & \checkmark & \checkmark & 17.1 & 16.4 & 3.50 & 0.886 & \\
\rowcolor{gray!10}
% 3a & \method-A-Max & \checkmark & \checkmark & Anechoic & \checkmark & \checkmark & \cellcolor{top1}{17.2} & \cellcolor{top1}{16.5} & \cellcolor{top3}{3.52} & \cellcolor{top3}{0.888} & \\
3a & \method-A-Max$^*$ & \checkmark & \checkmark & Anechoic & \checkmark & \checkmark & {17.2} & {16.5} & {3.52} & {0.888} & \\
3b & \method-D-Max$^*$ & \checkmark & \checkmark & Anechoic & \texttimes & \checkmark & 14.4 & 13.4& 3.19 & 0.860 & \\
\midrule
% 4a & \multirow{2}{*}{\shortstack{UBliSS-Diff-ML \\ (A,B)}} & \checkmark & \checkmark & Anechoic & \checkmark & \checkmark & 16.8 & 16.0 & 3.47 & 0.882 & \\
\rowcolor{gray!10}
4a & \method-A-ML & \checkmark & \checkmark & Anechoic & \checkmark & \checkmark & \cellcolor{top1}{16.9} & \cellcolor{top2}{16.2} & \cellcolor{top3}{3.49} & \cellcolor{top3}{0.884} & \\
4b & \method-D-ML & \checkmark & \checkmark & Anechoic & \texttimes & \checkmark & 14.0 & 12.9 & 3.15 & 0.851 & \\
\midrule
\rowcolor{gray!10}
5a & TF-GridNet-SMS$^\dagger$ & \texttimes & \texttimes & - & - & - &\cellcolor{top2}{16.8} & \cellcolor{top1}{16.3} & \cellcolor{top1}{3.91} & \cellcolor{top1}{0.924} & \\
5b & TF-GridNet-SMS & \texttimes & \texttimes & - & - & - & \cellcolor{top3}{16.2} & \cellcolor{top3}{15.7} & \cellcolor{top2}{3.72} & \cellcolor{top2}{0.908} & \\
\rowcolor{gray!10}
5c & TF-GridNet-Spatial & \texttimes & (3-mics) & - & - & - & 14.7 & 14.1 & 3.35 & 0.877 & \\
\bottomrule
\end{tabular}
\vspace{-13pt}
\end{table*}
\vspace{-8pt}
\section{Experiments and Evaluation}
\vspace{-3pt}
We first discuss {\method}'s training and sampling configurations, as well as the baseline configurations.
Then we show results and analysis on both \textbf{fixed} microphone array (test samples are all from one single microphone array) and \textbf{ad-hoc} microphone arrays (test samples are from different unknown microphone arrays).
% We have open sourced {\method}\footnote{ArrayDPS code in~\href{https://github.com/ArrayDPS/ArrayDPS}{https://github.com/ArrayDPS/ArrayDPS}}.
We have open sourced {\method} in~\href{https://github.com/ArrayDPS/ArrayDPS}{https://github.com/ArrayDPS/ArrayDPS}.
\vspace{-5pt}
\subsection{Datasets}\label{sec:datasets}
\vspace{-2pt}
As mentioned in Sec.~\ref{sec:method}, {\method} relies on a pre-trained single-channel single-speaker denoising diffusion model $D_\theta$. 
We train this unconditional speech diffusion model on a clean subset of speech corpus LibriTTS~\cite{libritts}. 
Since the virtual source could be reverberant, we train another diffusion prior to reverberant speech, where we convolve LibriTTS clean speech with room impulse responses (RIR). 
We discuss details, including the architecture and diffusion training configurations, in Appendix~\ref{app:diffusion}. 
For the final sampler and the posterior score approximation, we discuss the parameter configurations in Appendix~\ref{app:algorithm_param}.

% For microphone recordings, we first evaluate on the SMS-WSJ~\cite{smswsj} dataset, which contains simulated mixtures and sources on a fixed circular microphone array. 
% This is a common dataset/benchmark used for multi-channel speech separation~\cite{unssor, eras, smswsjuse, mix2mix, tfcrossnet}; details of the dataset are reported in Appendix~\ref{app:smswsj}. 
% We use the first, the third, and the fifth microphone for 3-channel speech separation, as in UNSSOR~\cite{unssor}.
For evaluation, we use SMS-WSJ~\cite{smswsj} dataset for fixed microphone array evaluation
% This is a common dataset/benchmark used for multi-channel speech separation~\cite{unssor, smswsjuse, mix2mix, tfcrossnet}.
and use Spatialized WSJ0-2Mix dataset~\cite{mcdeepcluster} for ad-hoc microphone array evaluation. SMS-WSJ contains simulated mixtures and sources on a fixed microphone array while Spatialized WSJ0-2Mix contains samples on different microphone arrays. All datasets are designed for 2-speaker source separation. Details of the datasets are explained in Appendix~\ref{app:dataset}.

\vspace{-5pt}
\subsection{\method~and Variants (for ablations)}\label{sec:arraydps_configs}
\vspace{-2pt}
Row 2a--2e, 3a--3b, and 4a--4b in Table~\ref{tab:smswsj} show \method~and its variants for ablations.
\textbf{\method-A} in row \textbf{2a} is default {\method} with anechoic clean speech diffusion prior. 
\textbf{\method-B} in row~\textbf{2b} is the same but the diffusion prior is trained on reverberant speech. 
\textbf{\method-C} in row \textbf{2c} shows the ablation when the reference channel guidance is removed from Algorithm~\ref{algo:score_cond}.
\textbf{\method-D} in row \textbf{2d} shows the ablation when IVA initialization is removed, while \textbf{\method-E} in row~\textbf{2e} shows the ablation when both initialization and reference channel guidance are removed. 
Note that the sampling parameters for \method~with and without initialization are different (see Appendix~\ref{app:algorithm_param}). 
Because all \method~in Row 2(a--e) are generative, we sample 5 separated samples for each mixture and report the 5-sample mean and standard deviation (mean$\pm$std in Row 2, Table~\ref{tab:smswsj}), averaged over all the mixtures in the test set.

In Row \textbf{3a--3b}, \textbf{\method-A-Max} and \textbf{\method-D-Max} shows the max metric score from the 5 samples, averaged for all test samples. 
Row~\textbf{4a-4b} represents \textbf{\method-A-ML} and \textbf{\method-D-ML}, where ML stands for `maximum likelihood'. 
Basically, from the 5 samples, we find the one with the maximum likelihood (highest mixture reconstruction SNR) and report the separation metric averaged over the test set. 
Row 3a and 4a's methods are the same as the default 2a (default), while 3b and 4b are the same as 2d (no initialization). 
Note that Row 4 is a practical algorithm because it only uses the mixture to calculate the likelihood.
% \vspace{-5pt}
\subsection{Baseline Methods\label{sec:baselines}}
% \vspace{-1pt}
For the baseline models, we include three unsupervised models: \textbf{Spatial Clustering}~\cite{wdo, spatialcluster1, spatialcluster2, spatialcluster3}, \textbf{Independent Vector Analysis (IVA)}~\cite{iva1, iva2}, \textbf{UNSSOR}~\cite{unssor}, as well as supervised TF-GridNet~\cite{tfgridnet}.
\textbf{TF-GridNet-SMS$^\dagger$} is the supervised baseline model reported in UNSSOR~\cite{unssor} trained on SMS-WSJ Corpus~\cite{smswsj}. \textbf{TF-GridNet-SMS} and \textbf{TF-GridNet-Spatial} are our reproduced TF-GridNet trained on SMS-WSJ and Spatialized WSJ0-2Mix corpus respectively. Note that TF-GridNet-Spatial in Table~\ref{tab:smswsj} row \textbf{5c} should generalize to any 3-channel ad-hoc microphone arrays.
% (note that TF-GridNet-Spatial is still trained on fixed number of microphones, so it is array agnostic when the number of microphone is fixed.). 
All baseline models are discussed in detail in Appendix~\ref{app:baseline}.

\vspace{-5pt}
\subsection{Metrics\label{sec:metrics}}
\vspace{-1pt}
For metrics, we use the first channel clean source images as the reference signals. 
We use the signal-to-distortion ratio (SDR)~\cite{sdr} and Scale-Invariant SDR (SI-SDR)~\cite{sisdr} to measure sample-level separation performance. 
We use perceptual evaluation of speech quality (PESQ)~\cite{pesq} and extended short-time objective intelligibility (eSTOI)~\cite{stoi} to measure perceptual quality and speech intelligibility respectively.
\vspace{-5pt}
\subsection{Fixed Microphone Array Evaluation}\label{sec:arrayspecific}
\vspace{-1pt}
We show the evaluation result for the 3-channel (using first, third, and fifth mics) SMS-WSJ test set in Table~\ref{tab:smswsj}. 6-channel SMS-WSJ results are in Appendix~\ref{app:result_smswsj}.

\textbf{Against Unsupervised:} 
% $\langle$$\rangle$
We first compare against the unsupervised methods.
Comparing rows \textbf{2a vs. 1(a--d)}, we see \method-A's mean score consistently outperforms spatial clustering and IVA-based methods by a substantive margin; 
also outperforms the unsupervised state-of-the-art method UNSSOR in all metrics except eSTOI. 
Moreover, {\method-A}'s standard deviation is quite robust. 
Note that UNSSOR is trained only for a fixed array on the SMS-WSJ dataset. 
% By observing our method's standard deviation, the performance is quite robust. 
If we further compare row \textbf{4a vs. 1d}, we see that by selecting the maximum likelihood sample, \method-A-ML shows strong improvement over UNSSOR (1.8 dB of SI-SDR, 0.29 PESQ).

\textbf{Against Supervised:} 
% Second, we compare our method with supervised method. 
Outperforming recent supervised methods is extremely challenging and all previous unsupervised methods exhibit a gap to supervised methods~\cite{unssor, eras}. 
On comparing row \textbf{2a vs. 5a--5b}, \method-A is consistently worse than the recent supervised methods, as expected.
However, comparing row \textbf{3a vs. 5a--5b}, \method-A-Max shows better SDR and SI-SDR than the supervised method, which means if we sample 5 samples using \method, one of the samples can outperform supervised methods in terms of SDR and SI-SDR.
However, since \method-Max's max operation is not practical, we check \method-A-ML, i.e., row \textbf{4a  vs. 5a--5b}; observe that the maximum likelihood sample from 5 samples achieves slightly better SDR than the supervised method, but a bit worse on other metrics. 
% Note that the multi-channel version of TF-GridNet is a recent carefully designed state of the art architecture, which our method cannot take advantage of.  (UNSSOR report supervised performance of other architectures like DPRNN~\cite{dprnn} and TCN-DenseUNet~\cite{tcn}, which our method can easily outperform).  

\textbf{Ablation Studies:}
We compare the ablations in rows 2a-2e. 
Comparing \textbf{2b vs. 2a}, we find that it is better to use the anechoic clean speech diffusion prior instead of the reverberant one. 
% One explanation is that it is easier for FCP to learn the actual RIR (transformation from anechoic to reverberant) instead of the relative RIR (transformation from reverberant to reverberant), which is further explained in Appendix~\ref{app:fcp}. 
The possible reason is related to FCP, which we explain in Appendix~\ref{app:fcp}. 
Comparing \textbf{2c vs. 2a}, we find that with initialization, the reference channel guidance can improve all metrics and reduce the standard deviation (instability). 
Comparing \textbf{2d vs. 2a}, the performance drops severely without the IVA initialization, while the std is also high, meaning sometimes the method can separate but not always. 
When it's able to separate, from row \textbf{3b} and \textbf{4b}, the best result from 5 samples is reasonable even without IVA initialization.
Lastly, we compare row \textbf{2d vs. 2e} and find that without IVA initialization, the reference channel guidance is extremely important to obtain reasonable performance.

Lastly, we check row \textbf{5c}, which is supervised TF-GridNet trained on 3-channel Spatialized WSJ0-2Mix dataset (ad-hoc array dataset)
% The dataset contains randomized microphone array geomrtries and will be explained in Appendix~\ref{app:spatialwsj}.
% In row 5c, TF-GridNet-Spatial \hl{should work for any 3-channel settings}, so we test its generalization performance on SMS-WSJ dataset. 
and we test how well it works on the SMS-WSJ dataset. Observing row \textbf{5b vs. 5c}, we see that for supervised models, generalizing to ad-hoc arrays costs a huge performance drop. Comparing row \textbf{5b vs. 2a}, \method-A outperforms supervised TF-GridNet-Spatial in SDR and SI-SDR by about 1dB, while \method-A-ML (row \textbf{4a}) outperforms TF-GridNet-Spatial for all metrics.
\begin{table}[t]
\vspace{-8pt}
\caption{Ad-hoc Array Evaluation results for 4-channel Spatialized WSJ0-2Mix. Note that the microphone positions are random for this dataset. Top results are emphasized in \colorbox{top1}{top1}, \colorbox{top2}{top2}, and \colorbox{top3}{top3}. Methods denoted with $*$ mean it is impractical.}
\vspace{-8pt}
\setlength\tabcolsep{4.5pt}
\renewcommand{\arraystretch}{0.8}  % Adjust line spacing for clarity
\label{tab:spatialwsj}
\vskip 0.15in
\begin{center}
\begin{small}
\begin{tabular}{@{}cccccc@{}}
\toprule
Row &Methods & SDR & SI-SDR & PESQ & eSTOI \\
\midrule
% \rowcolor{gray!10}
0 & Mixture    & 0.2 & 0.0 & 1.81 & 0.545 \\
\midrule
1a&Spatial Cluster    & 9.3 & 8.0 & 2.48 & 0.745\\
% \rowcolor{gray!10}
1b&IVA-Laplace & 7.9 & 5.2 & 2.41 & 0.648 \\
1c&IVA-Gaussian    & 12.5 & 10.1 & 3.01 & 0.808 \\
% \rowcolor{gray!10}
1d&UNSSOR    & 15.2 &14.2 &\cellcolor{top2}{3.54} & 0.873 \\
\midrule
\multirow{2}{*}{2a}&\multirow{2}{*}{\method-A}      & \cellcolor{top3}{\multirow{2}{*}{\shortstack{16.1\\$\pm$0.6}}} & \cellcolor{top3}{\multirow{2}{*}{\shortstack{15.3\\$\pm$0.6}}} & \multirow{2}{*}{\shortstack{3.47\\$\pm$0.07}} & \multirow{2}{*}{\shortstack{0.877\\$\pm$0.012}} \\
&& & & & \\
\midrule
\multirow{2}{*}{3a}&\multirow{2}{*}{\method-D}      & \multirow{2}{*}{\shortstack{6.7\\$\pm$4.5}} & \multirow{2}{*}{\shortstack{4.7\\$\pm$5.0}} & \multirow{2}{*}{\shortstack{2.45\\$\pm$0.51}} & \multirow{2}{*}{\shortstack{0.677\\$\pm$0.120}} \\
&& & & & \\
\midrule
% \rowcolor{gray!10}
4a&\method-A-Max*&\cellcolor{top1}{16.8}& \cellcolor{top1}{16.0} & \cellcolor{top2}{3.54} & \cellcolor{top2}{0.891}\\
4b&\method-D-Max*& 12.9 & 11.6 & 3.14 & 0.830 \\
\midrule
% \rowcolor{gray!10}
5a&\method-A-ML      & \cellcolor{top2}{16.6} & \cellcolor{top2}{15.8} & {3.52} & \cellcolor{top3}{0.886} \\
5b&\method-D-ML& 12.3 & 10.8 & 3.06 & 0.810 \\
\midrule
% \rowcolor{gray!10}
6a&TF-GridNet-Spatial & 15.8 & 15.1 & \cellcolor{top1}{3.67} & \cellcolor{top1}{0.895} \\
\bottomrule
\end{tabular}
\end{small}
\end{center}
\vskip -10pt
\end{table}
\vspace{-6pt}
\subsection{Ad-hoc Array Evaluation}\label{sec:arrayagnostic}
\vspace{-1pt}
% \vspace{-10pt}
We evaluate the ad-hoc array setting on the Spatialized WSJ0-2Mix~\cite{mcdeepcluster} dataset. In Table~\ref{tab:spatialwsj}, we show results for the 4-channel case (using the first 4 mics for each sample, different samples' array geometries are different); the 2-channel and 3-channel results are in Appendix~\ref{app:result_spatial}.

Following Table~\ref{tab:spatialwsj}, {\method}-A-ML (row \textbf{5a}) outperforms all {\em unsupervised} methods by a substantive margin, while {\method}-A (row \textbf{2a})'s mean metrics performs better than UNSSOR (row \textbf{1d}) in SDR, SI-SDR, and eSTOI. Comparing row \textbf{2a vs. 6a}, we see {\method-A}'s mean metric scores are even better than {\em supervised} TF-GridNet in SDR and SI-SDR, showing \method's superiority in ad-hoc array setting.
\vspace{-8pt}

\subsection{3 speakers and Real-World Evaluations}
We also show that our \method works nicely in 3-speaker and real-world samples. For 3-speaker evaluation, we evaluate \method on the Spatialized WSJ0-3Mix dataset, which is the 3-speaker version of the Spatialized WSJ0-2Mix dataset. This means the dataset also contains samples recorded from ad-hoc microphones. The separation results are shown in Appendix~\ref{app:3speakers}, showing that ArrayDPS easily outperforms supervised methods by a large margin.

For real-world mixture evaluation, we recorded 15 mixtures in a 7m x 4m x 2.7m room. Two volunteers speak simultaneously with weak environmental noise and are recorded by 3 microphones. The separation results are shown in the demo site: \href{https://arraydps.github.io/ArrayDPSDemo/}{https://arraydps.github.io/ArrayDPSDemo/}.
\vspace{-8pt}
\subsection{Source and Filter Visualizations}
Recall that \method samples virtual source signals and then uses FCP to estimate the relative RIRs that filter the virtual sources to the reference channel sources. Thus, we conduct a visualization analysis of the virtual sources, and the estimated relative RIRs are in the spectrogram domain. The visualization and a detailed explanation are in Appendix~\ref{app:visualization}. The virtual source listening demos are also on our demo site. Our observation is that the virtual source separated is closer to the anechoic clean source speech than the reverberant clean source speech, meaning that \method has some dereverberation effects.
\vspace{-5pt}
\section{Related Work}
\vspace{-1pt}
% mixit~\cite{mixit, mixit2, mixit3, mixit4, mixit5, mixit6}

% mcmixit~\cite{mcmixit}

% arrayagnostic~\cite{arrayagnostic1, arrayagnostic2}

% ica~\cite{ica1, ica2, ica3}

% iva~\cite{iva1, iva2, iva_iss}

% % plabel~\cite{plabel1, plabel2, plabel3, plabel4, plabel5}

% unsup~\cite{unsup1, neuralfca, neuralfca2}

% % cocktail party problem, classical algorithms (ICA, IVA)
% Blind speech separation has been advanced greatly by deep neural networks using supervised training~\cite{wang_supervised_2018}. To solve the source label permutation issue, two types of work are deep clustering~\cite{deepcluster} and permutation invariant training~\cite{pit}. However, supervised methods have generalization issues~\cite{generalization1, generalization2}, especially in multi-channel scenarios, where training on multiple microphone with array agnostic tricks~\cite{arrayagnostic1, arrayagnostic2, arrayagnostic3} cannot guarantee generalization on unseen microphone arrays~\cite{arrayagnostic3}.
Blind speech separation has been advanced greatly by deep neural networks using supervised training~\cite{wang_supervised_2018, deepcluster, pit, tfcrossnet, tfgridnet, spatialnet, mimocsm, fovnet, tabe}. However, supervised methods have generalization issues~\cite{generalization1, generalization2}, especially in multi-channel scenarios, where training on multiple microphones with array agnostic tricks~\cite{arrayagnostic1, arrayagnostic2, arrayagnostic3} cannot guarantee generalization on unseen microphone arrays~\cite{arrayagnostic3}. Also, supervised methods are usually discriminative, often producing blurred results with nonlinear distortions. To solve this problem, generative separation methods~\cite{gan_sep, diff_sep, tornike, tornike2, sep_diff} are proposed to sample clean sources conditioned on the given mixture. However, these generative methods still need clean-mixture paired data for supervised training.

For unsupervised separation, independent component analysis (ICA)~\cite{ica1, ica2, ica3} separates sources by enforcing independence between sources. Independent vector analysis (IVA)~\cite{iva1, iva2, iva_iss} adds a Gaussian assumption on each source's STFT bins. Spatial Clustering~\cite{spatialcluster1, spatialcluster2, spatialcluster3, wdo} based methods separate the sources by clustering the STFT bins using spatial features.

Deep learning based unsupervised source separation has also been studied widely. Mixture invariant training (MixIT)~\cite{mixit, mixit2, mixit3, mixit4, mixit5, mixit6} synthesizes new mixtures by mixing real-world mixtures, and train the separation model output sources can be mixed somehow to the original mixtures. MixIT has been modified to the multi-channel scenarios~\cite{unssor, mcmixit}, but shows limited performance. Another line of work~\cite{plabel1, plabel2, plabel3, plabel4, plabel5} uses classic methods like spatial clustering to generate pseudo-labels as training targets. However, the performance is bounded by the classic method. More recently, one group of work tried to exploit the multi-channel signal model for unsupervised training~\cite{ras, eras,unssor,unsup1, neuralfca, neuralfca2}. These methods train a neural network to separate sources, which after filtering and mixing, would be close to the multi-channel mixture.
% Neural-FCA~\cite{neuralfca} trains this in a mixture VAE framework to make sure the sources are independent, while UNSSOR~\cite{unssor} nicely designs a mixture-constraint loss and an intra-source magnitude scattering loss. 
Although these methods show promising results, they do not exploit the speech prior information during training.

One way to incorporate prior is to use a diffusion model~\cite{ddpm, edm, score, sohl2015deep}. Score-based diffusion model learns to generate from a distribution by following a probabilistic SDE/ODE from a noise initialization~\cite{score}.
% At each diffusion step, it estimates the score function, and then update a small step towards the score direction, which will eventually converges to a data sample.
Diffusion posterior sampling (DPS)~\cite{DPS, pigdm, score} tries to solve the inverse problem with a pre-trained prior diffusion model
% During sampling, DPS tries to estimate the posterior score with original diffusion model
and a known likelihood model. Later, a few works are further proposed to solve the blind inverse problem with DPS~\cite{blinddps, buddy, bai2024blindinversionusinglatent, vsvento2025estimation, eloi1, eloi2, eloi3, fastdiffusionem}.

More recently, DPS has been applied for source separation~\cite{dpssep1, msdm, ssdgp}. However, these methods only consider single-channel source separation where the likelihood is tractable, which is not the case for multi-channel problems.
\vspace{-8pt}
\section{Conclusion}
\vspace{-1pt}
% This paper proposes {\method}, an unsupervised, array-agnostic, generative method to separate speech from multi-microphone mixture recordings. {\method} proposes to use a speech diffusion prior and a novel likelihood approximation to enable posterior sampling, which facilitates separation. The result shows our method outperforms other unsupervised methods and 
% performs on par with supervised methods, especially in ad-hoc microphone array settings.
This paper proposes {\method}, an unsupervised, array-agnostic, and generative method to separate speech from multi-channel mixture recordings. {\method} proposes to use a speech diffusion prior and a novel likelihood approximation to enable posterior sampling. The result shows our method outperforms other unsupervised methods and 
performs on par with supervised methods in SDR. We leave future work to explore \method~for more general array inverse problems.
\section*{Acknowledgment}
We thank Foxconn and NSF (grant 2008338, 1909568, 2148583,
and MRI-2018966) for funding this research. We are also grateful to the reviewers for their insightful feedback.

\section*{Impact Statement}
This paper explores advancements in unsupervised, array-agnostic, and generative blind speech separation, which lies in the intersection of several established fields, including array signal processing and machine learning. The algorithm we developed has many positive societal consequences. One notable application of our algorithm is assistive/augmented hearing, which improves communication and accessibility for individuals in noisy environments, especially those with hearing impairments. Also, our method has the potential to improve automatic speech transcription technology on ad-hoc microphone arrays. Although our method is generative, it does not synthesize harmful speeches that are not initially presented in the speech mixture. Nonetheless, the capability to isolate individual voices may pose privacy risks, which need careful regulations. We believe no other concerns require specific emphasis at this point.

\bibliography{main}
\bibliographystyle{icml2025}

%%%%%%%%%%%%%%%%%%%%%%%%%%%%%%%%%%%%%%%%%%%%%%%%%%%%%%%%%%%%%%%%%%%%%%%%%%%%%%%
%%%%%%%%%%%%%%%%%%%%%%%%%%%%%%%%%%%%%%%%%%%%%%%%%%%%%%%%%%%%%%%%%%%%%%%%%%%%%%%
% APPENDIX
%%%%%%%%%%%%%%%%%%%%%%%%%%%%%%%%%%%%%%%%%%%%%%%%%%%%%%%%%%%%%%%%%%%%%%%%%%%%%%%
%%%%%%%%%%%%%%%%%%%%%%%%%%%%%%%%%%%%%%%%%%%%%%%%%%%%%%%%%%%%%%%%%%%%%%%%%%%%%%%
\newpage
\appendix
\onecolumn
We organize the appendix as follows:
\begin{itemize}
    \item Appendix~\ref{app:diff} provides a background on score-based diffusion and diffusion posterior sampling (DPS).
    \item  Appendix~\ref{app:posterior_score} provides proposed posterior score derivation, analysis, and relation to the classic EM algorithm.
    \item Appendix~\ref{app:algorithms} presents Algorithm~\ref{alg:ubliss} and Algorithm~\ref{algo:score_cond} in detail, and provides algorithm configurations.
    \item Appendix~\ref{app:diffusion} explains diffusion prior training configuration and model architecture.
    \item Appendix~\ref{app:baseline} discusses all the baseline models in detail.
    \item Appendix~\ref{app:dataset} illustrates the SMS-WSJ and Spatialized WSJ0-2Mix datasets for evaluation.
    \item Appendix~\ref{app:fcp} shows ablation experiments for FCP method and explains why we use the virtual source signal model.
    \item Appendix~\ref{app:result_spatial} reports the result of 2-channel and 3-channel Spatialized WSJ0-2Mix dataset.
    \item Appendix~\ref{app:result_smswsj} reports the result of 6-channel SMS-WSJ dataset.
    \item Appendix \ref{app:3speakers} reports the 3-speaker separation results on the 4-channel Spatialized WSJ0-3Mix dataset.
    \item Appendix \ref{app:ablations} shows ArrayDPS's sensitivity and robustness to hyperparameters.
    \item Appendix \ref{app:visualization} shows the visualization of ArrayDPS separated virtual sources, FCP estimated RIRs, and ArrayDPS separated reference channel sources.
\end{itemize}

\section{Score-based Diffusion and Diffusion Posterior Sampling}\label{app:diff}
{\bf Score Diffusion Model:}
Diffusion-based generative models aim to sample from a data distribution $p_{\text{data}}(s)$, by starting from a noise sample and then gradually denoising it~\cite{score, edm, ddpm}. Score-based diffusion model~\cite{score, edm} can also formulate the diffusion denoising process with a probabilistic flow ODE.  Following EDM~\cite{edm}, with the specific diffusion noise schedule $\sigma(\tau)=\tau$, where $\tau$ is the diffusion time, then the probabilistic flow ODE is defined by:
{\small
\begin{equation}\label{eq:ode}
d{s}^\tau = -\sigma(\tau) \nabla_{s^\tau} \log p(s^\tau) d\tau    
\end{equation}
}
where $p(s^\tau) =  \mathcal{N}(s, \sigma^2(\tau)I)$. This probabilistic flow ODE allows transformation from $s^{\tau_{\text{max}}}$ (noise) to $s^{\tau_{\text{min}}}$ (data), using an ODE solver like Euler's method. During training, a denoiser $D_\theta(s^\tau, \sigma(\tau))$ is learned to denoise $s^\tau$ with the following denoising loss:
% \begin{equation}\label{eq:denoise_diff}
% \mathbb{E}_{\tau \sim \mathcal{U}(\tau_0,\tau_{\text{max}}) }\mathbb{E}_{s \sim p_{\text{data}}}\mathbb{E}_{n \sim \mathcal{N}(0, \sigma^2(\tau) I)} \left\| D_\theta(s + n; \sigma) - s \right\|^2
% \end{equation}
{\small
\begin{equation}\label{eq:denoise_diff}
\mathbb{E}_{\tau}\mathbb{E}_{s \sim p_{\text{data}}}\mathbb{E}_{n \sim \mathcal{N}(0, \sigma^2(\tau) I)} \left\| D_\theta(s + n,~\sigma(\tau)) - s \right\|^2
\end{equation}
}
After training, the score function $\nabla_{{s}^\tau} \log p({s}^\tau)$ is approximated by $S_{\theta}(s^\tau, \sigma(\tau)) =: \frac{D_\theta(s^\tau, \sigma(\tau)) - s^\tau}{\sigma^2(\tau)}$, using the Tweedie's Formula. During sampling, $s^{\tau_{\text{max}}}$ is first sampled from $\mathcal{N}(0, \sigma^2(\tau_{\text{max}})I)$, and then an ODE solver is used to integrate through~Eq.\ref{eq:ode} from $s^{\tau_{\text{max}}}$ to a data sample $s^{\tau_{\text{min}}}\sim p_{\text{data}}$. 

{\bf Diffusion Posterior Sampling:} This diffusion model further allows a universal framework to solve the general inverse problem, which is known for diffusion posterior sampling~\cite{DPS, pigdm}. Consider a clean signal $s$ (e.g., image or speech) is distorted by a distortion function $A(\cdot)$, which results in a distorted signal $x = A(s) + \epsilon_n$. To recover the original clean signal $s$, diffusion posterior sampling proposes to sample from $p(s|x)$ by using the probabilistic flow ODE similar to Eq.~\ref{eq:ode} except using the posterior score: %initiated from $s^{\tau_{max}}\sim \mathcal{N}(0, \sigma^2(\tau_{max})I)$:
{\small
\begin{equation}\label{eq:dpsode}
d{s}^\tau = -\sigma(\tau) \nabla_{s^\tau} \log p(s^\tau | x) d\tau    
\end{equation}
}
To approximate the posterior score $\nabla_{s^\tau} \log p(s^\tau | x)$, it's decomposed into a prior score and a likelihood score with Bayes' theorem:
{\small
\begin{equation}\label{eq:bayesian}
    \nabla_{s^\tau} \log p(s^\tau | x) = \nabla_{s^\tau} \log p(s^\tau) + \nabla_{s^\tau} \log p(x|s^\tau)
\end{equation}
}
where $\nabla_{s^\tau} \log p(s^\tau)$ is estimated by $S_\theta(s^\tau, \sigma(\tau))$, and in DPS~\cite{DPS}, the second term $\nabla_{s^\tau} \log p(x|s^\tau)$ is empirically approximated by:
{\small
\begin{align}
    &\nabla_{s^\tau} \log p(x|\hat{s}^0(s^\tau)) \text{, where }\hat{s}^0(s^\tau)=D_\theta(s^\tau,~\sigma(\tau)) \label{eq:dps_app1}\\
    &=\nabla_{s^\tau}\xi(\tau)\|x-A(D_\theta(s^\tau,~\sigma(\tau)))\|_2^2 \label{eq:dps_app2}
\end{align}
}
The $\xi(\tau)$ is a parameter to control the amount of likelihood guidance for each step, which theoretically relates to the measurement noise variance. Note that this solution assumes that the distortion function $A(\cdot)$ is known in advance.

% In this paper, we consider the case where $s$ represents the virtual channel speech sources $s_{1:K}$ and the distorted measurement $x$ represent the multi-channel mixture $x_{1:C}$,  as mentioned in Sec.~\ref{sec:problem}.

\section{Posterior Score Approximation}\label{app:posterior_score}
This section first gives a derivation of the proposed posterior score approximation as in Sec.~\ref{sec:score_cond} Eq.~\ref{eq:final_score}. Then to understand why this makes sense, we give a thorough analysis of our approximation in Sec.~\ref{app:analysis}, where we also prove FCP as a maximum likelihood estimator. Finally, in Sec.~\ref{app:em}, we connect our likelihood score approximation step with the classic expectation maximization algorithm.
\subsection{Derivation}\label{app:derivation}
This section gives a full detailed derivation of our approximation of the posterior score as proposed in Sec.~\ref{sec:score_cond}. Basically, our goal is to approximate $\nabla_{s^\tau_{1:K,1}}\log p(s^\tau_{1:K} | x_{1:C})$ that's needed for sampling. First, similar to Eq.~\ref{eq:bayesian}, we decompose $\nabla_{s^\tau_{1:K}}\log p(s^\tau_{1:K} | x_{1:C})$ with the Bayes' theorem:
% \begin{align}
%     &\nabla_{s^\tau_{1:K}} \log p(s^\tau_{1:K} | x_{1:C}) \nonumber\\
%     = &\nabla_{s^\tau_{1:K}} \log p(s^\tau_{1:K}) + \nabla_{s^\tau_{1:K}} \log p(x_{1:C}|s^\tau_{1:K})\\
%     = &\sum_{k=1}^{K}\nabla_{s^\tau_{1:K}} \log p(s^\tau_{k}) + \nabla_{s^\tau_{1:K}} \log p(x_{1:C}|s^\tau_{1:K}) \label{eq:ourbayesian}
%     % \simeq&\sum_{k=1}^{K}\nabla_{s^\tau_{1:K}} \log p(s^\tau_{k}) + \nabla_{s^\tau_{1:K}} \log p(x_{1:C}|s^\tau_{1:K}) \label{eq:ourbayesian_approx}
% \end{align}
\begin{align}
    \nabla_{s^\tau_{1:K}} \log p(s^\tau_{1:K} | x_{1:C})
    = &\nabla_{s^\tau_{1:K}} \log p(s^\tau_{1:K}) + \nabla_{s^\tau_{1:K}} \log p(x_{1:C}|s^\tau_{1:K})\\
    = &\sum_{k=1}^{K}\nabla_{s^\tau_{k}} \log p(s^\tau_{k}) + \nabla_{s^\tau_{1:K}} \log p(x_{1:C}|s^\tau_{1:K}) \label{eq:ourbayesian}
    % \simeq&\sum_{k=1}^{K}\nabla_{s^\tau_{1:K}} \log p(s^\tau_{k}) + \nabla_{s^\tau_{1:K}} \log p(x_{1:C}|s^\tau_{1:K}) \label{eq:ourbayesian_approx}
\end{align}
Now notice that in the first part of Eq.~\ref{eq:ourbayesian}, all the summands $\nabla_{s^\tau_{k}} \log p(s^\tau_{k})$ (prior score) can be approximated by the score model $S_\theta(s^\tau_{k},~\sigma(\tau))$ trained on single-speaker speech. Thus, we try to approximate the second part of Eq.~\ref{eq:ourbayesian} (likelihood score): %similar to Eq.~\ref{eq:dps_app1}:
% \begin{align}
%     &\nabla_{s^\tau_{1:K}} \log p(x_{1:C}|s^\tau_{1:K})\\
%     =&\sum_{c=1}^{C}\nabla_{s^\tau_{1:K}} \log p(x_{c}|s^\tau_{1:K}) \label{eq:likelihood1}\\
%     % =&\nabla_{s^\tau_{1:K,1}} \log p(x_{1}|s^\tau_{1:K,1}) + \sum_{c=2}^{C}\nabla_{s^\tau} \log p(x_{c}|s^\tau_{1:K,1}) \label{eq:likelihood1}\\
%     \simeq&\sum_{c=1}^{C}\nabla_{s^\tau_{1:K}} \log p(x_{c}|\hat{s}^0_{1}(s^\tau_{1}), ..., \hat{s}^0_{K}(s^\tau_{k}))\text{, where } \hat{s}^0_{k}(s^\tau_{k})=D_\theta(s^\tau_{k}, \sigma(\tau))\label{eq:likelihood2}
% \end{align}
\begin{align}
    \nabla_{s^\tau_{1:K}} \log p(x_{1:C}|s^\tau_{1:K})=&\sum_{c=1}^{C}\nabla_{s^\tau_{1:K}} \log p(x_{c}|s^\tau_{1:K}) \label{eq:likelihood1}\\
    % =&\nabla_{s^\tau_{1:K,1}} \log p(x_{1}|s^\tau_{1:K,1}) + \sum_{c=2}^{C}\nabla_{s^\tau} \log p(x_{c}|s^\tau_{1:K,1}) \label{eq:likelihood1}\\
    \simeq&\sum_{c=1}^{C}\nabla_{s^\tau_{1:K}} \log p(x_{c}|\hat{s}^0_{1}(s^\tau_{1}), ..., \hat{s}^0_{K}(s^\tau_{K}))\text{,~~~~~where } \hat{s}^0_{k}(s^\tau_{k})=D_\theta(s^\tau_{k}, \sigma(\tau))\label{eq:likelihood2}
\end{align}
Eq.~\ref{eq:likelihood1} is correct because we assume that for the mixtures, different channels' mixtures are conditionally independent given all the sources, which is our array-agnostic assumption. Eq.~\ref{eq:likelihood1} to Eq.~\ref{eq:likelihood2} is based on the likelihood approximation in DPS, as shown before in Eq.~\ref{eq:dps_app1}.  However, in Eq.~\ref{eq:likelihood2}, $\log p(x_{c}|\hat{s}^0_{1}(s^\tau_{1}), ..., \hat{s}^0_{K}(s^\tau_{K}))$ is still not tractable because the estimated clean sources are single-channel virtual sources, while the mixtures are multi-channel. To build the connection, we transfer to the STFT domain and then estimate the relative RIRs to transform from virtual sources to real channel source images.

First, follow the notation in Sec.~\ref{sec:fcp}, we denote $X_{c}$ as the STFT of $x_c$ and denote $\hat{S}^0_{k}$ as the STFT of $\hat{s}^0_{k}(s^\tau_{k})$. Then we use the FCP algorithm defined in Sec.~\ref{sec:fcp}  Eq.~\ref{eq:fcp_func} and Eq.~\ref{eq:fcp} to estimate the relative RIRs which allows the transformation from virtual source estimates to all-channel source images:
% \begin{align}
%     \hat{G}^{k}_{0\rightarrow c} &= \text{FCP}(X_{c}, \hat{S}^{0,\tau}_{k}) \label{eq:fcp_filter_gen}\\
%     \hat{S}_{k,c}(l, f) &= \hat{G}^{k}_{0\rightarrow c}(l, f) * \hat{S}_{k}(l, f) \label{eq:fcp_filter}
% \end{align}
\begin{align}
    \hat{G}^{k}_{0\rightarrow c} &= \text{FCP}(X_{c}, \hat{S}^{0}_{k})\text{, 
          }\hat{S}_{k,c}= \hat{G}^{k}_{0\rightarrow c} *_l \hat{S}^0_{k},~~~k\in\{1,2,...,K\},~~~c\in\{1,2,...,C\}\label{eq:fcp_filter}
\end{align}
Now given the estimated source images $\hat{S}_{k,c}$, we are able to approximate the likelihood score in Eq.~\ref{eq:likelihood2} in a tractable way:
\begin{align}
    &\nabla_{s^\tau_{1:K}} \log p(x_{c}|\hat{s}^0_{1}(s^\tau_{1}), ..., \hat{s}^0_{K}(s^\tau_{K}))\label{eq:likelihood_nonref1}\\
        = &\nabla_{s^\tau_{1:K}} \log p(X_{c}|\hat{S}^{0}_{1}, ..., \hat{S}^{0}_{K})\label{eq:likelihood_nonref2} \\
    \simeq &\nabla_{s^\tau_{1:K}} \log p(X_{c}|\hat{S}^{0}_{1}, ..., \hat{S}^{0}_{K}, \hat{G}^{1}_{0\rightarrow c}, ..., \hat{G}^{K}_{0\rightarrow c})\label{eq:likelihood_nonref3} \\
    = &\nabla_{s^\tau_{1:K}}\xi(\tau)\left\|X_{c}-\sum_{k=1}^{K}\left( \hat{G}^{k}_{0\rightarrow c}(l, f) *_l \hat{S}^0_{k}(l, f)\right)\right\|_2^2\label{eq:likelihood_nonref4}\\
    = &\nabla_{s^\tau_{1:K}}\xi(\tau)\left\|X_{c}-\sum_{k=1}^{K}\hat{S}_{k,c}\right\|_2^2\label{eq:likelihood_nonref5}\\
    \simeq &\nabla_{s^\tau_{1:K}}\xi(\tau)\left\|x_c-\text{ISTFT}\left(\sum_{k=1}^{K}\hat{S}_{k,c}\right)\right\|_2^2\label{eq:likelihood3}
\end{align}
Eq.~\ref{eq:likelihood_nonref1} to Eq.~\ref{eq:likelihood_nonref2} is just a STFT transform. Eq.~\ref{eq:likelihood_nonref2} to Eq.~\ref{eq:likelihood_nonref3} is an empirical approximation. A simple intuition of this is that Eq.~\ref{eq:likelihood_nonref3} is tractable while Eq.~\ref{eq:likelihood_nonref2} is not. The reason is that the relative RIRs ${G}^{k}_{0\rightarrow c}$ are needed for the likelihood to be tractable, according to the signal model in Eq.~\ref{eq:spec_model}. Thus an empirical choice is to estimate the ${G}^{k}_{0\rightarrow c}$ using FCP, as in Eq.~\ref{eq:fcp_filter}. Later in next section, we will show that these FCP estimated relative RIRs $\hat{G}^{1:K}_{0\rightarrow {1:C}}$ are exactly maximum likelihood relative RIRs: 
\begin{equation}\label{eq:gmap}
    \hat{G}^{1}_{0\rightarrow c}, .., \hat{G}^{K}_{0\rightarrow c} = \underset{ {G}^{1}_{0\rightarrow c}, .., {G}^{K}_{0\rightarrow c}}{\arg\max} p(X_{c}|\hat{S}^{0}_{1}, .., \hat{S}^{0}_{K}, {G}^{1}_{0\rightarrow c}, .., {G}^{K}_{0\rightarrow c})
\end{equation}
We analyze the validity from Eq.~\ref{eq:likelihood_nonref2} to Eq.~\ref{eq:likelihood_nonref3} (likelihood score approximation) in the next subsection where we also show the connection to the classic Expectation Maximization (EM) algorithm. Eq.~\ref{eq:likelihood_nonref3} to Eq.~\ref{eq:likelihood_nonref4} is based on the signal model in Eq.~\ref{eq:spec_model} with the Gaussian measurement noise. Eq.~\ref{eq:likelihood_nonref4} to Eq.~\ref{eq:likelihood_nonref5} is a simplification using Eq.~\ref{eq:fcp_filter}, and Eq.~\ref{eq:likelihood_nonref5} to Eq.~\ref{eq:likelihood3} is based on the power preservation of STFT or the STFT version of Parserval's theorem.

Now we can finalize our derivation of the approximated posterior score  using Eq.~\ref{eq:ourbayesian}, Eq.~\ref{eq:likelihood2}, and Eq.~\ref{eq:likelihood3}, which derives our result in Eq.~\ref{eq:final_score} Sec.~\ref{sec:score_cond}:
\begin{align}
    \nabla_{s^\tau_{1:K}} \log p(s^\tau_{1:K} | x_{1:C})\simeq\sum_{k=1}^{K}S_\theta(s^\tau_{k}, \sigma(\tau))+\sum_{c=1}^{C}\nabla_{s^\tau_{1:K}}\xi(\tau)\left\|x_c-\text{ISTFT}\left(\sum_{k=1}^{K}\hat{S}_{k,c}\right)\right\|_2^2 \label{eq:final_score2}
\end{align}
\subsection{Analysis and Validation}\label{app:analysis}
This section complements the derivation above showing the analysis and validation of our posterior score approximation. We first show that the relative RIR filter estimation using FCP is equivalent to maximum likelihood filter estimation, which would validate Eq.~\ref{eq:gmap}. We then validate the likelihood score approximation from Eq.~\ref{eq:likelihood_nonref2} to Eq.~\ref{eq:likelihood_nonref3}. Lastly, we show the likelihood score approximation's relationship with the classic Expectation Maximization algorithm.

To show the FCP as a maximum likelihood relative RIR estimator, we first recap the FCP formulation in Eq.~\ref{eq:fcp_func} and Eq.~\ref{eq:fcp_obj}, 
\begin{align}
    % \hat{\lambda}^k_c(l,f) &= \text{some expression defining lambda} \\
    \hat{G}^k_{0\rightarrow c}(l, f) &= \text{FCP} (X_c(l, f), \hat{S}_{k}(l, f))=\underset{G^k_{0\rightarrow c}(l, f)}{\arg\min} \sum_{l,f} \frac{1}{\hat{\lambda}^k_c(l,f)} \left|X_c(l, f) - G^k_{0\rightarrow c}(l, f)*\hat{S}_{k}(l, f)\right|^2\label{eq:fcp_func2}
    % &\hspace{4cm}+\epsilon\left\|G^k_{1,c}\right\|_2^2\\
    % &\hat{\lambda}^k_c(l,f) = 
\end{align}
\begin{theorem}\label{thm1}
Assume $X_c(l, f) =\sum\limits_{k=1}^{K} G^k_{0\rightarrow c}(l, f)*_l S_{k}(l, f) + N_c(l, f)$, where $N_c(l,f)\sim\mathcal{N}(0, \sigma_N^2I)$. Then when  $\hat{\lambda}^k_c(l,f)=\frac{1}{2\sigma_N^2}$, the FCP relative RIR estimator is a maximum likelihood filter estimator:
\begin{equation}
        \text{FCP}(X_c, \hat{S}_1^0), ..., \text{FCP}(X_c, \hat{S}_K^0)=\underset{{G}^{1}_{0\rightarrow c}, .., {G}^{K}_{0\rightarrow c}}{\arg\max} \log p(X_{c}| {G}^{1}_{0\rightarrow c}, .., {G}^{K}_{0\rightarrow c}, \hat{S}^{0}_{1}, .., \hat{S}^{0}_{K})
\end{equation}
\end{theorem}
\begin{proof}
First, let $\hat{G}^{1}_{0\rightarrow c}, .., \hat{G}^{K}_{0\rightarrow c}$ be the maximum likelihood estimator as shown below:
\begin{align}
    \hat{G}^{1}_{0\rightarrow c}, .., \hat{G}^{K}_{0\rightarrow c} &= \underset{{G}^{1}_{0\rightarrow c}, .., {G}^{K}_{0\rightarrow c}}{\arg\max} \log p(X_{c}| {G}^{1}_{0\rightarrow c}, .., {G}^{K}_{0\rightarrow c}, \hat{S}^{0}_{1}, .., \hat{S}^{0}_{K})\label{eq:fcp_ml1}\\
    &= \underset{{G}^{1}_{0\rightarrow c}, .., {G}^{K}_{0\rightarrow c}}{\arg\max} \log \mathcal{N}(X_{c}; \sum\limits_{k=1}^{K} {G}^{K}_{0\rightarrow c} *_l \hat{S}^{0}_{K}, \sigma^2_NI)\label{eq:fcp_ml2}\\
    &= \underset{{G}^{1}_{0\rightarrow c}, .., {G}^{K}_{0\rightarrow c}}{\arg\min} \frac{1}{2\sigma^2_N} \left\|X_{c} - \sum\limits_{k=1}^{K} {G}^{K}_{0\rightarrow c} *_l \hat{S}^{0}_{K}\right\|_2^2\label{eq:fcp_ml3}
\end{align}
Eq.~\ref{eq:fcp_ml1} to Eq.~\ref{eq:fcp_ml3} is based on the signal model in the assumption. Then based on the orthogonal principle in estimation theory, independence of the sources, and linearity of the filtering operation, Eq.~\ref{eq:fcp_ml3} is equivalent to:
\begin{equation}
    \hat{G}^{k}_{0\rightarrow c} = \underset{{G}^{k}_{0\rightarrow c}}{\arg\min} \frac{1}{2\sigma^2_N} \left\|X_{c} - {G}^{k}_{0\rightarrow c} *_l \hat{S}^{0}_{k}\right\|_2^2~~~~~~~,\forall k \in \{1, 2, ..., K\} \label{eq:fcp_ml4}
\end{equation}
which is exactly the FCP objective mentioned in Eq.~\ref{eq:fcp_func2}, where $\hat{\lambda}^k_c(l,f)=\frac{1}{2\sigma_N^2}$ as in the assumption. 
\end{proof} With Theorem~\ref{thm1}, it is also obvious that FCP is equivalent to the maximum posterior solution of the relative RIR because we do not assume any prior of the relative RIR filter. Note that empirically in FCP, the weight $\hat{\lambda}^k_c(l,f)$ is set to $\hat{\lambda}^k_c(l,f) = \frac{1}{C} \sum_{c=1}^C |X_{c}(l, f)|^2+ \epsilon\cdot \max\limits_{l,f}\frac{1}{C} \sum_{c=1}^C |X_{c}(l,f)|^2$ for better regularization.

Now with Theorem~\ref{thm1} relating FCP with filter likelihood, we further analyze the correctness of the likelihood score approximation from Eq.~\ref{eq:likelihood_nonref2} to Eq.~\ref{eq:likelihood_nonref3} in the derivation section~\ref{app:derivation}.   From Eq.~\ref{eq:likelihood_nonref2} to Eq.~\ref{eq:likelihood_nonref3}, we have the FCP estimated relative RIRs so that the likelihood is tractable:
\begin{align}
   \nabla_{s^\tau_{1:K}} \log p(X_{c}|\hat{S}^{0}_{1}, ..., \hat{S}^{0}_{K})
    \simeq \nabla_{s^\tau_{1:K}} \log p(X_{c}|\hat{S}^{0}_{1}, ..., \hat{S}^{0}_{K}, \hat{G}^{1}_{0\rightarrow c}, ..., \hat{G}^{K}_{0\rightarrow c})\label{eq:likelihood_approx}
\end{align}
Now we find the assumption that would make the approximation into equality:
\begin{align}
    p&(X_{c}|\hat{S}^{0}_{1}, ..., \hat{S}^{0}_{K})= \int_{{G}^{1}_{0\rightarrow c}, ..., {G}^{K}_{0\rightarrow c}}p({G}^{1}_{0\rightarrow c}, ..., {G}^{K}_{0\rightarrow c})p(X_{c}|\hat{S}^{0}_{1}, ..., \hat{S}^{0}_{K}, {G}^{1}_{0\rightarrow c}, ..., {G}^{K}_{0\rightarrow c})d{G}^{1}_{0\rightarrow c}, ..., d{G}^{K, \tau}\\
    &=p(X_{c}|\hat{S}^{0}_{1}, ..., \hat{S}^{0}_{K}, \hat{G}^{1}_{0\rightarrow c}, ..., \hat{G}^{K}_{0\rightarrow c}) \text{~~~~if~~} p(X_{c}|\hat{S}^{0}_{1}, ..., \hat{S}^{0}_{K}, {G}^{1}_{0\rightarrow c}, ..., {G}^{K}_{0\rightarrow c})= 
    \begin{cases}
        1 \text{~~~~if~~} {G}^{k}_{0\rightarrow c}=\hat{G}^{k}_{0\rightarrow c} \text{ for all } k\\
        0 \text{~~~~else}
    \end{cases}
\end{align}
We can see that the approximation becomes equality under the assumption that $p(X_{c}|\hat{S}^{0}_{1}, ..., \hat{S}^{0}_{K}, {G}^{1}_{0\rightarrow c}, ..., {G}^{K}_{0\rightarrow c})$ is a delta function of all the relative RIRs ${G}^{1}_{0\rightarrow c}, ..., {G}^{K}_{0\rightarrow c}$, where the function has non-zeros values only when the relative RIRs are the FCP estimated relative RIRs $\hat{G}^{1}_{0\rightarrow c}, ..., \hat{G}^{K}_{0\rightarrow c}$, which are also the maximum likelihood estimators as mentioned before. This assumption makes sense when there is no measurement noise, where a slight error in the relative RIR would cause a likelihood of 0. However, of course, when there is relatively larger measurement noise, the assumption is no longer true, but would still work empirically as shown in our SMS-WSJ separation result in Table~\ref{tab:smswsj}.

\subsection{Likelihood Approximation and Expectation Maximization}\label{app:em}
While we analyze the correctness of our posterior score approximation in the previous section, we also find that the most challenging step (likelihood score approximation~Eq.~\ref{eq:likelihood_approx}) has a strong connection with the classic Expectation Maximization algorithm. 
% We then start to relate our likelihood approximation (Eq.~\ref{eq:likelihood_approx} or Eq.~\ref{eq:likelihood_nonref3}) with the classic Expectation Maximization (EM) algorithm. 
Note that in EM algorithm, when the likelihood $p(X_{c}|\hat{S}^{0}_{1}, ..., \hat{S}^{0}_{K})$ is intractable because of the unobserved relative RIRs ${G}^{1}_{0\rightarrow c}, ..., {G}^{K}_{0\rightarrow c}$, the Expectation (E) step and the Maximization (M) step are iterated:
\begin{align}
    &\text{E-Step: calculate }\psi({G}^{1}_{0\rightarrow c}, ..., {G}^{K}_{0\rightarrow c})\gets p({G}^{1}_{0\rightarrow c}, ..., {G}^{K}_{0\rightarrow c} | X_{c}, \hat{S}^{0}_{1}, ..., \hat{S}^{0}_{K})\label{eq:estep}\\
    &\text{M-Step: update }\hat{S}^{0}_{1}, ..., \hat{S}^{0}_{K}\gets \underset{\hat{S}^{0}_{1}, ..., \hat{S}^{0}_{K}}{\arg\max}~\mathbb{E}_{\psi({G}^{1}_{0\rightarrow c}, ..., {G}^{K}_{0\rightarrow c})}\left[\log~p(X_{c}|\hat{S}^{0}_{1}, ..., \hat{S}^{0}_{K}, {G}^{1}_{0\rightarrow c}, ..., {G}^{K}_{0\rightarrow c})\right]\label{eq:mstep}
\end{align}
The E step updates the distribution $\psi({G}^{1}_{0\rightarrow c}, ..., {G}^{K}_{0\rightarrow c})$ of the unobserved relative RIRs to be the distribution of the relative RIRs given the current source estimates $\hat{S}^{0}_{1}, ..., \hat{S}^{0}_{K}$ and the mixture $X_{c}$. Then using the E step updated distribution $\psi({G}^{1}_{0\rightarrow c}, ..., {G}^{K}_{0\rightarrow c})$, the M step maximizes the expectation of log likelihood by update the source estimates $\hat{S}^{0}_{1}, ..., \hat{S}^{0}_{K}$. Then this procedure is iterated until convergence, which is called the Expectation Maximization algorithm.

In our case, we only want to estimate the current log likelihood without updating the sources estimate $\hat{S}^{0}_{1}, ..., \hat{S}^{0}_{K}$, so our log likelihood approximation only contains one E step and then calculate the expected log likelihood as in the M step (Eq.\ref{eq:mstep}). However, The E step is also not tractable in our case, so we just approximate it with:
\begin{align}
     \psi({G}^{1}_{0\rightarrow c}, ..., {G}^{K}_{0\rightarrow c})&\simeq
 \begin{cases}
    1 \text{~~if  } {G}^{k}_{0\rightarrow c}=\hat{G}^{k}_{0\rightarrow c} \text{ for all } k\\
    0 \text{~~else}
\end{cases}\\
\text{where }\hat{G}^{1}_{0\rightarrow c}, \hat{G}^{2, \tau}_{0\rightarrow c}, ..., \hat{G}^{K}_{0\rightarrow c}&=\underset{{G}^{1}_{0\rightarrow c}, .., {G}^{K}_{0\rightarrow c}}{\arg\max} \log p({G}^{1}_{0\rightarrow c}, .., {G}^{K}_{0\rightarrow c}|X_{c}, \hat{S}^{0}_{1}, .., \hat{S}^{0}_{K})\label{eq:em1}\\
&=\underset{{G}^{1}_{0\rightarrow c}, .., {G}^{K}_{0\rightarrow c}}{\arg\max} \log p(X_{c}| {G}^{1}_{0\rightarrow c}, .., {G}^{K}_{0\rightarrow c}, \hat{S}^{0}_{1}, .., \hat{S}^{0}_{K})\label{eq:em2}\\
&=\text{FCP}(X_c, \hat{S}^{0}_{1}),~\text{FCP}(X_c, \hat{S}^{0}_{2}), ...,~\text{FCP}(X_c, \hat{S}^{0}_{K})\label{eq:em3}
\end{align}
We set the distribution $\psi({G}^{1}_{0\rightarrow c}, ..., {G}^{K}_{0\rightarrow c})$ in the E step to be deterministic at $\hat{G}^{1}_{0\rightarrow c}, ..., \hat{G}^{K}_{0\rightarrow c}$, which are both the maximum likelihood filters and maximum posterior filters because no priors are assumed for the filters.
Then as mentioned in Sec.~\ref{app:analysis} Theorem~\ref{thm1}, $\hat{G}^{1}_{0\rightarrow c}, ..., \hat{G}^{K}_{0\rightarrow c}$ can be directly estimated by the FCP method. Then after the E step, we can calculate the expected log likelihood in M step by only calculating $\log~p(X_{c}|\hat{S}^{0}_{1}, ..., \hat{S}^{0}_{K}, \hat{G}^{1}_{0\rightarrow c}, ..., \hat{G}^{K}_{0\rightarrow c})$, because $\psi$ is defined to be a deterministic distribution.

\section{Algorithm Details}\label{app:algorithms}
This section provides all the details of the two algorithms in Sec.~\ref{sec:algo}, including \method~as in Algorithm~\ref{alg:ubliss} and posterior score approximation in Algorithm~\ref{algo:score_cond}.
\subsection{\method~Algorithm}\label{app:algorithm_arraydps}
This section discusses the \method~Algorithm~\ref{alg:ubliss} in detail.\\
{\bf IVA Initialization:}
From lines 1-6 in Algorithm~\ref{alg:ubliss}, IVA first separates $K$ speech source images in the reference channel (channel 1). Then these IVA-separated sources are used as source estimates to calculate the relative RIR using FCP. Note that in lines 4-6, the relative RIR $\hat{G}^{\text{{IVA}}}_{1\rightarrow c}$ is calculated for all channels. These IVA-initialized relative RIRs will be later used for posterior score approximation as in Algorithm~\ref{algo:score_cond}, which will be discussed later. Then line 7-9 initializes the starting diffusion noise with the IVA separated sources, in the reference channel (channel 1). Although we are trying to output a virtual source, we still initialize it to be the reference channel's IVA outputs.

{\bf Stochastic Sampler:}
Lines 10-24 in Algorithm~\ref{alg:ubliss} show the discrete sampler that integrates through the separation ODE in Eq.~\ref{eq:ode_method_overview} starting from the diffusion noise initialized in line 8. This sampler is proposed in EDM~\cite{edm}, with the parameters $N, \sigma_{i\in\{0,...,N-1\}}, \gamma_{i\in\{0,...,N-1\}}, S_{\text{noise}}\}$ and inputs $\{x_{1:C}, K\}$. $x_{1:C}$ is the multi-channel mixture input and $K$ is the number of sources to separate. $N$ is the number of diffusion denoising steps. $\sigma_{i\in\{0,...,N\}}$ is the noise schedule for each diffusion denoising step (i.e., amount of noise to remove at each step), where $\sigma_N = 0, \sigma_0=\sigma(\tau_{max}), \sigma_{N-1}=\sigma(\tau_{\text{min}})$. Then all other steps' schedule follows:
    \begin{align}\label{eq:sigma}
    \sigma_{i < N} = \left( \sigma(\tau_\text{max})^{\frac{1}{\rho}} + \frac{i}{N-1} \left( \sigma(\tau_\text{min})^{\frac{1}{\rho}} - \sigma(\tau_\text{max})^{\frac{1}{\rho}} \right) \right)^{\rho} \quad
    \end{align}
$\rho$ is a parameter that controls how the noise level per step changes from $\sigma(\tau_{max})$ to $\sigma(\tau_{min})$, e.g., $\rho > 1$ means the noise level per step will decrease more rapidly as $i$ increases. Note that in our ODE settings as in Sec.\ref{sec:method}, $\sigma(\tau_i)=\tau_i=\sigma_i$. $\gamma_{i\in\{1,...,N-1\}}$ are parameters to control the amount of stochasticity added in the beginning of each step, as in lines 11-14. In EDM~\cite{edm}, $\gamma_{i\in\{1,...,N-1\}}$ is set to be:
\begin{align}\label{eq:gamma}
    \gamma_i = 
    \begin{cases} 
    \min \left( \frac{S_{\text{churn}}}{N}, \sqrt{2} - 1 \right) & \text{if } \sigma_i \in [S_{\text{min}}, S_{\text{max}}] \\
    0 & \text{otherwise}
    \end{cases}
\end{align}
, where $S_{\text{churn}}$ controls the amount of stochasticity and $S_{\text{min}}, S_{\text{max}}$ set the steps for adding stochasticity. $S_{\text{noise}}$ in line 12 also controls the amount of stochasticity within step $i\in[S_{\text{min}}, S_{\text{max}}]$. In general, the sampler first has a usual $1^{\text{st}}$ order Euler step to update the sources (16-18), and then the newly updated source is used to calculate a $2^{\text{nd}}$ order correction step (lines 20-22). More details of the sampler can be checked in the original EDM paper~\cite{edm}.

{\bf Post FCP Filtering:}
Note the sampler mentioned above in Algorithm~\ref{sec:algo} is separating virtual sources (line 24). To get the final separated source images in reference channel 1, post-filtering is needed. Lines 26-29 use the FCP algorithm to estimate the relative RIR that can transform the virtual sources into the reference channel. Then after convolutional filtering (line 28), the final reference channel-separated source images are returned as outputs (line 30).

% {\bf Posterior Score Estimation:}
\subsection{Posterior Score Approximation Algorithm}
In Algorithm~\ref{alg:ubliss}, each update step needs to calculate the posterior score, which uses the \texttt{GET\_SCORE} function in Algorithm~\ref{algo:score_cond}. The function \texttt{GET\_SCORE}($s^{\tau_i}_{1:K}$, $x_{1:C}$, $\hat{G}^{k,\text{IVA}}_{1\rightarrow~c}$, $i$) aims to approximate the posterior score $\nabla_{s^\tau_{1:K}} \log p(s^\tau_{1:K} | x_{1:C})$ at sampling step $i$, using our propose approximation formulation in Sec.~\ref{sec:score_cond}, but with some empirical modifications. 
% $s^\tau_{0,1:K}$ and $x_{1:C}$ are obviously the inputs to the function, while the IVA initialized relative RIRs $\hat{G}^{k,\text{IVA}}_{0, c}$ and the sampler step $i$ are also needed.
The function needs a few pre-defined parameters $\{D_\theta, \sigma_{i\in\{0,...,N\}}, N_{\text{ref}}, N_{\text{fg}}, \xi_1(\tau), \xi_2(\tau)\}$. $D_\theta$ is the diffusion denoiser mentioned in Sec.~\ref{sec:score_cond}. $\sigma_{i\in\{0,...,N\}}$ is the sampler's noise schedule as mentioned in the stochastic sampler paragraph. 

$N_{\text{ref}}$ is a sampler step threshold such that when $i\leq N_{\text{ref}}$, there is an empirical reference-channel guidance term that aims to match the sum of estimated virtual sources $\sum_{k=1}^{K}\hat{s}_k$ to the reference channel (channel 1) mixture $x_1$. This is shown in line 18 in Algorithm.~\ref{algo:score_cond}. We found that this is important to the algorithm's stability. $N_{\text{fg}}$ is also a sampler step threshold for the filtering guidance of IVA initialized relative RIR. As shown in lines 7-14, when $i\leq N_{\text{fg}}$, instead of using FCP to estimate the relative RIR, the IVA initialized relative RIR $\hat{G}^{k,\text{IVA}}_{1\rightarrow~c}$ is directly used. We found that this helps a lot for sampling stability and performance gain. The intuition is that in the first $N_{\text{fg}}$ steps, the source estimate would be extremely blurred and so as the relative RIR estimates, while the IVA initialized relative RIR is a much better choice. Note that $\hat{G}^{k,\text{IVA}}_{1\rightarrow~c}$ are actually relative RIRs from reference channel 1 to all channels, but our estimated sources are all virtual sources. We clarify here that this is on purpose, which means in the first $N_{\text{fg}}$ steps, the virtual source estimation would be close to the reference channel, and then the constraint is relaxed in later steps.

$\xi_1(\tau)$ and $\xi_2(\tau)$ are the weights for the mixture likelihood guidance in lines 16 and l8 in Algorithm.~\ref{algo:score_cond}. Theoretically, $\xi_1(\tau)=\frac{1}{2\sigma_n^2}$, where $\sigma^2_n$ is the white noise variance defined in Eq.~\ref{eq: waveform_mix}, and $\xi_2(\tau)=0$, because the term in line 18 is empirical. However, for sampling stability, we resort to the common practice used in previous audio posterior sampling work~\cite{cqt, bwe, buddy}, where we set $\xi_1(\tau)$ and $\xi_2(\tau)$ to the following:
\begin{align}
    \xi_1(\tau) &= \frac{\xi \sqrt{T}}{\tau\left\|\nabla_{s^\tau_{1:K}}\sum\limits_{c=1}^{C}\left\|x_c-\sum_{k=1}^{K}\hat{s}_{k,c}\right\|_2^2\right\|_2}\label{eq:xi1}\\
    \xi_2(\tau) &= \frac{\xi \sqrt{T}}{\tau\left\|\nabla_{s^\tau_{1:K}}\left\|x_1-\sum_{k=1}^{K}\hat{s}_{k}\right\|_2^2\right\|_2}\label{eq:xi2}
\end{align}
All the variables in Eq.~\ref{eq:xi1} and Eq.~\ref{eq:xi2} directly refer to the same variables in lines 16 and 18 in Algo.~\ref{algo:score_cond}. $T$ is the length of the audio samples. $\xi$ is just a single scalar parameter to tune the mixture likelihood guidance. Note that $\xi_1(\tau)$ and $\xi_2(\tau)$ directly depend on the calculated gradient term so the the notations $\xi_1(\tau)$ and $\xi_2(\tau)$ are not accurate. However, we stick to it because the notation is simpler and it is also used in previous papers~\cite{cqt, bwe}.

Overall, in Algorithm~\ref{algo:score_cond}, the diffusion denoiser first denoises the noisy virtual sources (lines 3), and then the prior score is calculated using Tweedie's formula (line 6). Then depending on the current sample step $i$, the algorithm either chooses the IVA initialized relative RIR (line 12) or directly estimates the relative RIR (line 10) using FCP. The relative RIR is then used to filter the virtual source estimates to all real channels (channel 1 to C) (line 14), which allows the calculation of the gradient likelihood guidance (line 16) and the reference channel guidance (19). Finally, the prior score and the guidance terms are added together to output the approximated posterior score (line 21).

\subsection{Algorithm Configurations}\label{app:algorithm_param}
Here we show our \method~configurations. For the default configuration as in row 2a (\method-A) in Table~\ref{tab:smswsj}, we set $N=400$ and $S_{\text{noise}}=1$ as in Algorithm~\ref{alg:ubliss}, $\sigma_0=\tau_{\text{max}}=0.8$, $\sigma_N=\tau_{\text{min}}=1e-6$, $\rho=10$ as in Eq.~\ref{eq:sigma}, $S_{\text{min}}=0$, $S_{\text{max}}=50$, and $S_{\text{churn}}=30$ as in Eq~\ref{eq:gamma}, $\xi=2$, $N_{\text{ref}}=200$, $N_{\text{fg}}=100$, and $\lambda=1.3$ as in Algorithm~\ref{algo:score_cond}. The diffusion prior trained on anechoic speech is used. For the FCP algorithm, we use an FFT size of 512 (64 ms), hop size of 64 (8 ms), square root Hanning window, and $\epsilon=0.001$ as in Eq.~\ref{eq:fcp_lambda}. For the IVA initialization, we use the open-source \textit{torchiva} toolkit~\cite{torchiva}, and we use FFT size of 2048 (256 ms), hop size 256 (32 ms), Gaussian Prior, and 100 iterations. 

For \method-B in row 2b of Table~\ref{tab:smswsj}, the diffusion model is trained on reverberant speech, as discussed in Sec.~\ref{sec:arraydps_configs} and Appendix~\ref{app:diffusion}. For \method-C in row 2c of Table~\ref{tab:smswsj}, the reference channel guidance in Algorithm~\ref{algo:score_cond} is removed, which means $\lambda=0$. For \method-D and \method-E in row 2d-2e of Table~\ref{tab:smswsj}, the IVA initialization is removed, which means $N_{\text{fg}}=0$ and no IVA is needed. Moreover, in the case of no IVA initialization, we set $\sigma_0=\tau_{\text{max}}=2$ and $\xi=6$ since we found this set of parameters perform better when tuning on validation sets. For \method-E, the reference-channel guidance is also removed, so $N_{\text{ref}}=0$ in this case.

\section{Diffusion Training Details}\label{app:diffusion}
% (need to first mention training datasets)
As mentioned in Sec.\ref{sec:arraydps_configs}, the diffusion prior is trained on the train-clean-100 and train-clean-360 subset of LibriTTS dataset~\cite{libritts}. This subset contains about 460 hours clean speech with 1000+ speakers. However, this training set is primarily reverberation-free or anechoic. Since the virtual source could be reverberant, we also train a second diffusion model on reverberant speech. The reverberant speech dataset is the same subset of LibriTTS, but each speech utterance is convolved with a randomly selected room impulse response (RIR) to synthesize reverberant speech. The randomly selected RIRs are from the SLR26 and SLR28 dataset~\cite{slr}, which contains 3,076 real and about 115,000 synthetic RIRs. 

For the diffusion denoising architecture, we use the waveform domain U-Net as MSDM~\cite{msdm}, implemented in audio-diffusion-pytorch/v0.0.432\footnote{\href{https://github.com/archinetai/audio-diffusion-pytorch/tree/v0.0.43}{https://github.com/archinetai/audio-diffusion-pytorch/tree/v0.0.43}}. The U-Net takes audio waveform as inputs, and also outputs audio waveforms. The U-Net consists of encoder, bottleneck, and decoder. The encoder contains 6 layers of 1-D convolutional ResNet blocks, where the last three layers also contain multi-head slef-attention blocks following the ResNet block. The input channel is 1 because we are modeling single channel speech. The number of output channels for each encoder layer is [256, 512, 1024, 1024, 1024, 1024], and the number of downsampling factors for each layer is [4, 4, 4, 2, 2, 2]. For self attention blocks, we use 8 attention heads where each is 128 dimensional.
The bottleneck also contains an attention block, with one ResNet block before the attention block and one ResNet block after the attention block. The decoder is then reverse symmetric to the encoder, with output channel size 1 to output a same size waveform.

For diffusion training configurations, we use the EDM~\cite{edm} training framework following the training recipe in CQT-Diffusion\footnote{\href{https://github.com/eloimoliner/CQTdiff}{https://github.com/eloimoliner/CQTdiff}}. The diffusion denoiser $D_\theta(s^\tau,\sigma(\tau))$ is set to be a linear combination of the noisy input and the U-Net $f_\theta$ output:
\begin{align}
    D_\theta(s^\tau,\sigma(\tau)) = c_{\text{skip}}(\sigma(\tau))\cdot s^\tau + C_{\text{out}}(\sigma(\tau)) \cdot f_\theta(c_{\text{in}}s^\tau, c_{\text{noise}})\\
\end{align}
,where $c_\text{skip}(\sigma)$, $c_\text{out}(\sigma)$, $c_\text{in}(\sigma)$, $c_\text{noise}(\sigma)$, are defined as:
\begin{equation}
    c_{\text{skip}}(\sigma) = \frac{\sigma^2_{\text{data}}}{\sigma^2 + \sigma^2_{\text{data}}},\;
    c_{\text{out}}(\sigma) = \frac{\sigma\cdot\sigma_{\text{data}}}{\sqrt{\sigma^2 + \sigma^2_{\text{data}}}},\;
    c_{\text{in}}(\sigma) = \frac{1}{\sqrt{\sigma^2 + \sigma^2_{\text{data}}}},\;
     c_{\text{noise}}(\sigma) = \frac{1}{4} \ln(\sigma)
\end{equation}
$\sigma_{\text{data}}$ shown above is a parameter to set based on the standard deviation of the training dataset. We set this value to 0.057. Then the diffusion denoising objective is shown as:
\begin{equation}\label{eq:denoise_diff}
\mathbb{E}_{\tau \sim p_\tau }\mathbb{E}_{s \sim p_{\text{data}}}\mathbb{E}_{n \sim \mathcal{N}(0, \sigma^2(\tau) I)} \left\| D_\theta(s + n, \sigma(\tau)) - s \right\|^2
\end{equation}
where $p_\tau$ is empirically set to match the diffusion sampling scheduler's density. During training, we set $\sigma(\tau_{\text{max}})=10$, $\sigma(\tau_{\text{min}})=1e-6$, $\rho=10$, $S_{\text{churn}}=5$ for the diffusion scheduler. We train on speech samples with 65,536 samples ($\sim$8.2 s) with batsh size 16 and learning rate 0.0001. The learning rate is multiplied by 0.8 every 60,000 training steps. Also exponential moving average (EMA) is used to update the neural network weights during training. The EMA updated weights after training is used for sampling in inference time. We train the model for 840,000 training steps.

\section{Evaluation Datasets}\label{app:dataset}
In this paper, two datasets are used for evaluation. The first dataset is SMS-WSJ~\cite{smswsj} for \textbf{fixed} microphone array evaluation, where all the methods are evaluated on one single microphone array recorded samples. The second dataset is Spatialized WSJ0-2Mix~\cite{mcdeepcluster} for \textbf{ad-hoc} microphone array evaluation, where methods are evaluated on multiple unknown microphone arrays.
\subsection{SMS-WSJ Dataset}\label{app:smswsj}
SMS-WSJ~\cite{smswsj} is a commonly used reverberant speech separation corpus~\cite{unssor, smswsjuse, eras, tfcrossnet}. The dataset is simulated from WSJ0 and WSJ1 datasets. The dataset is for multi-channel 2-speaker separation in reverberant conditions for a fixed simulated microphone array. The simulated microphone array is a circular array with six microphones, uniformly on a circle with 10 cm of radius. For each sample, the shoebox room is randomly sampled, with the length and width uniformly sampled from [7.6, 8.4] m and [5.6, 6.4] m, respectively. The array center position is randomly sampled from [3.6, 4.4] m from the shorter wall and [2.6, 3.4] m from the longer wall. Then the array is randomly rotated along all three geometric axes. The mixture contains two speaker sources, which are sampled form WSJ0 and WSJ1 datasets. For each speaker, the speaker-to-array distance is uniformly sampled from [1.0, 2.0] m, and each speaker position's azimuth related to the array center is also randomly sampled. The room impulse responses are simulated using image-source method~\cite{imagesource}. The sound decay time (T60) is uniformly sampled from [200, 500] ms. Finally, to simulate sensor noise, an SNR value is sampled to be [20, 30] dB, and then the white noise is sampled and scaled to satisfy the SNR required. The dataset consists 33,561 ($\sim$87.4 h), 982 ($\sim$2.5 h), and 1,332 ($\sim$3.4 h) train, validation, and test mixtures, respectively, all in 8kHz sampling rate.

% SMS-WSJ [67] is a popular corpus for evaluating two-speaker separation algorithms in reverberant
% conditions. The clean speech is sampled from the WSJ0 and WSJ1 datasets. The corpus contains
% 33, 561 (∼87.4 h), 982 (∼2.5 h), and 1, 332 (∼3.4 h) two-speaker mixtures respectively for training,
% validation, and testing. The simulated microphone array has six microphones arranged uniformly on
% a circle with a diameter of 20 cm. For each mixture, the speaker-to-array distance is drawn from the
% range [1.0, 2.0] m, and the reverberation time (T60) is sampled from [0.2, 0.5] s. A weak white noise
% is added to simulate microphone self-noises, and the energy level between the sum of the reverberant
% speech signals and the noise is sampled from the range [20, 30] dB. The sampling rate is 8 kHz.
\subsection{Spatialized WSJ0-2Mix Dataset}\label{app:spatialwsj}
Spatialized WSJ0-2Mix~\cite{mcdeepcluster} is the spatialized version of the WSJ0-2Mix~\cite{deepcluster} dataset, which is a commonly used speech separation dataset mixed by randomly selecting sources from the WSJ0 corpus. Multi-channel Deep Clustering~\cite{mcdeepcluster} creates the Spatialized WSJ0-2Mix dataset for multi-channel reverberant speech separation. For each dataset sample in WSJ0-2Mix, a shoebox room, 8-channel microphone array geometry, source positions, and microphone array center are sampled randomly to spatialize that sample. The room length $l$, width $w$, and height are uniformly sampled from [5, 10] m, [5, 10] m, and [3, 4] m, respectively. The RIRs' sound decay time (T60) is sampled from [200, 600] ms. The microphone array center is set to be $(l/2+n_1, w/2+n_2, h)$, where $n_1$ and $n_2$ are uniformly sampled from [-0.2, 0.2] m, and $h$ is uniformly sampled from [1, 2] m. Then the microphone array geometry is randomly sampled. First, the array radius is uniformly sampled from [7.5, 12,5] cm. The first two microphones are sampled on the sphere with the sampled radius, and are symmetric according to the array origin. Two more microphones are randomly sampled inside the sphere while making sure the distances between any two microphones has to be at least 5 cm. Then 4 more microphones are randomly sampled inside the sphere without any restrictions. In total, eight microphones are randomly sampled. For the speakers, each speaker's location $(x, y, z)$ is sampled such that $l/2+n_1-1.5\leq x\leq l/2+n_1+1.5, w/2+n_2-1.5\leq y\leq w/2+n_2+1.5, 1.5\leq z\leq2$. The source locations are sampled such that the source-array distance is at least 0.5 m, and source-source distance is at least 1 m. The RIR simulation is also using the image source method~\cite{imagesource}. The sample rate is set to 8kHz.

In general, the Spatialized WSJ0-2Mix dataset contains 20,000($\sim$30h), 5,000 ($\sim$10h), and 3,000 ($\sim$5h) utterances in training, validation, and testing, respectively. Different utterances are using different microphone geometries, allowing training models that can generalize to diverse ad-hoc microphone arrays. Also, the direct-to-reverb energy radio of the dataset is 2.5 dB with 3.8 dB of standard deviation.

\section{Baseline Methods}\label{app:baseline}
As in Row \textbf{1a} of Table~\ref{tab:smswsj}, Spatial clustering~\cite{wdo, spatialcluster1, spatialcluster2, spatialcluster3} tries to cluster the multi-channel spatial features in STFT domain. Then each STFT bin is assigned to a source based on the clustering, which relies on the assumption that different speech do not overlap in STFT domain, known as W-Disjoint Orthogonality (W-DO)~\cite{wdo}. We use the same spatial clustering configuration as in the baselines of UNSSOR~\cite{unssor}, which uses the open source implementation~\cite{spatialclustergithub}. The FFT size is 1024 (256 ms) and the hop size is 128 (32 ms). The number of sources is set to $K+1$ and the lowest energy source is discarded as a garbage source.

For the Independent Vector Analysis baseline, similar to the baseline in UNSSOR~\cite{unssor}, we use the \textit{torchiva} toolkit with FFT size 2048 (64 ms) and hop size 256 (32 ms). In UNSSOR, the spherical Laplace prior is used as in \textbf{Row 1b} of Table~\ref{tab:smswsj}. We also include Gaussian prior because we find it to perform better, as shown in \textbf{Row 1c} of Table~\ref{tab:smswsj}.

For the UNSSOR baseline trained on Spatialized WSJ0-2Mix datasets, we use the default configuration where both MC and ISMS loss are used, and $I=19$ and $J=0$, following the notation in the paper. We use the open source code\footnote{\hyperlink{https://github.com/merlresearch/reverberation-as-supervision}{https://github.com/merlresearch/reverberation-as-supervision}} to train UNSSOR. We found that UNSSOR training is quite sensitive to the ISMS loss weight, which is $\gamma$ in the original paper. Thus we sweep $\gamma\in\{0.03, 0.04, 0.05, 0.06, 0.08, 0.3\}$ and report the best result for the 2/3/4-channel of Spatialized WSJ0-2Mix dataset. For the SMS-WSJ dataset, we just borrow the result in the original UNSSOR paper and report as UNSSOR$^\dagger$ as in Table~\ref{tab:smswsj}.

\textbf{Row 5a-5c} of Table~\ref{tab:smswsj} show the supervised TF-GridNet~\cite{tfgridnet} with permutation invariant training (PIT)~\cite{pit}, where TF-GridNet-SMS$^\dagger$ is the exact model trained on SMS-WSJ dataset provided as the supervised baseline of UNSSOR, showing superior performance. However, we also need a supervised baseline model trained on the Spatialized WSJ0-2Mix dataset. Thus we reproduce a supervised TF-GridNet training pipeline using the same architecture with TF-GridNet-SMS$^\dagger$ as in UNSSOR, but trained with with SI-SDR\cite{sisdr} PIT loss for simplicity (we are not clear about the training details of the TF-GridNet supervised baseline reported in UNSSOR). We use ADAM optimizer and start with 1e-4 learning rate. We half the learning rate whenever the validation loss do not improve in 3 consecutive epochs. We use batch size 4 and each training segment is 4 seconds long. As in \textbf{Row 5b} of Table~\ref{tab:smswsj}, TF-GridNet-SMS is our implemented supervised TF-GridNet trained on 3-channel SMS-WSJ with 120 epochs (about 1,000,000 steps). In Table~\ref{tab:spatialwsj}, TF-GridNet-Spatial is our implementation of the supervised TF-GridNet trained on Spatialized WSJ0-2Mix dataset mentioned in~\ref{app:spatialwsj}, which tries to generalize to any 3-channel microphone arrays. The model is trained for 200 epochs (about 1,000,000 steps) on Spatialized WSJ0-2Mix dataset. These models are all trained on a single A100 GPU and converges in about 5-6 days.

\section{FCP Ablations}\label{app:fcp}
The FCP mentioned in Sec.~\ref{sec:fcp} takes a $c_1$ channel mixture $X_{c_1}$ and a $c_2$ channel source estimate $\hat{S}_{k, c_2}$ as inputs, and then output the relative RIR $\hat{G}^k_{c_2\rightarrow c_1}$ that can transform $\hat{S}_{k, c_2}$ to $\hat{S}_{k, c_1}$:
\begin{align}
    \hat{G}^k_{c_2\rightarrow c_1} = \text{FCP}(X_{c_1}, \hat{S}_{k, c_2}),~~~~\hat{S}_{k, c_1} = \hat{G}^k_{c_2\rightarrow c_1} *_l \hat{S}_{k, c_2}\label{eq:fcp222}
\end{align}
As discussed in Sec.~\ref{sec:problem}, we decided to use the virtual sources as the estimated sources, which is more flexible comparing with reference channel source images. Here we list a few obvious choices for the virtual source and reason why we use the virtual source based on the oracle FCP performance.

First, the virtual source can be the anechoic source signal $\Tilde{S}_k$. Second, the virtual source can be a partially-reverberant source signal ${S}^{\text{early}}_{k, c_2}$ in channel $c_2$, where only the first 50 milliseconds of the actual room impulse response is applied to filter the anechoic source $\Tilde{S}_k$. Note that this early reverberant signal is motivated by the REVERB challenge~\cite{reverb}, where people found adding the first 50 ms of early reverberation to the anechoic signal sounds better. Third and last, since the virtual source is flexible, it can also be the channel $c_2$'s reverberant image ${S}_{k,c_2}$.

Now we do a few FCP ablations for different kinds of virtual sources. Without loss of generality, we let $c_1=1$. Remember the FCP function always takes two signal inputs: the FCP target signal, and the source estimate signal. In Eq.~\ref{eq:fcp222}, the FCP target is $X_{c_1}$, and the source estimate signal is the channel $c_2$ estimate $\hat{S}_{k, c_2}$. Here, we do the ablations for different kinds of FCP input combinations, where the source estimate inputs can be the three kinds of virtual sources mentioned above (anechoic, early, reverberant), and the FCP target can be either channel $c_1=1$ mixture, or the channel $c_1=1$ reverberant source image. 
Now given the ground-truth virtual sources and the FCP targets, FCP is used to estimate the relative RIRs. Then after applying the estimated relative RIR on the virtual source estimate, we get the channel 1 source image estimate $\hat{S}_{k,1}$ which is used to calculate the speech metrics in Table~\ref{tab:fcp_ablation}. We calculate all the results in SMS-WSJ test set.
\begin{table}[ht]
\centering
\caption{FCP ablations for different kinds of FCP inputs combinations.The metrics reported are on channel 1 source images.}
\renewcommand{\arraystretch}{1.3}  % Adjusts the height of the table rows
\label{tab:fcp_ablation}
\begin{tabular}{cccccccc}
\midrule
Row & Source Estimate & FCP Target & FCP & SDR (dB) & SI-SDR (dB) & PESQ & eSTOI \\\midrule
\rowcolor{gray!10}
1 & anechoic $\Tilde{S}_k$ & mixture (ch-1) $X_1$ & $\text{FCP}(X_1, \Tilde{S}_k)$ & 22.0 & 19.8 & 4.15 & 0.974 \\\midrule
2 & early (ch-2) $S^{\text{early}}_{k, 2}$ & mixture (ch-1) $X_1$ & $\text{FCP}(X_1, S^{\text{early}}_{k, 2})$ & 16.5 & 13.2 & 3.91 & 0.943 \\\midrule
\rowcolor{gray!10}
3 & reverb (ch-2) ${S}_{k,2}$ & mixture (ch-1) $X_1$ & $\text{FCP}(X_1, {S}_{k,2})$ & 14.7 & 11.3 & 3.59 & 0.917 \\\midrule
4 & anechoic $\Tilde{S}_k$ & reverb (ch-1) $S_{k,1}$ & $\text{FCP}(S_{k,1}, \Tilde{S}_k)$ & 34.9 & 33.3 & 4.45 & 0.997 \\\midrule
\rowcolor{gray!10}
5 & early (ch-2) $S^{\text{early}}_{k, 2}$& reverb (ch-1) $S_{k,1}$ & $\text{FCP}(S_{k,1}, S^{\text{early}}_{k, 2})$ & 18.2 & 14.5 & 4.15 & 0.967 \\\midrule
6 & reverb (ch-2) ${S}_{k,2}$ & reverb (ch-1) $S_{k,1}$ & $\text{FCP}(S_{k,1}, {S}_{k,2})$ & 16.0 & 12.1 & 3.84 & 0.941 \\\midrule
\end{tabular}
% \label{tab:audio_results}
\end{table}

First, observe the first three rows in Table~\ref{tab:fcp_ablation}, where the FCP target is set to be the channel 1 mixture. In rows 1, 2, and 3, the source estimate is set to be the anechoic source, the early reverberant source, and the reverberant source, respectively. We can see that when the source estimate is set to be the anechoic source, the FCP filtering performance is much better than when the source estimate is early reverberant or reverberant. To see why this is the case, remember that in our signal model in Eq.~\ref{eq:relative_rir} in Sec.~\ref{sec:problem}, the relative RIR from one channel source image to another channel source image involves a inverse filtering operation, which may not exist. This also explains why in row 2a-2b in Table~\ref{tab:smswsj}, using an anechoic speech prior performs better than using the reverberant speech prior.

Row 4,5,6 in Table~\ref{tab:fcp_ablation} shows the case when the FCP target is the channel 1 source image. Note that in row 4, when the source estimate is set to groud-truth anechoic source, the metrics is almost perfect, showing the correctness of the signal model used in Sec.~\ref{sec:problem}.
% Row 4,5,6 in Table~\ref{tab:fcp_ablation} shows the case when the FCP target is the channel 1 source image. Note that in row 4, when the source estimate is set to ground-truth anechoic source, the metrics is almost perfect, showing the correctness of the signal model used in Sec.~\ref{sec:problem}.
\section{2 and 3 Channels Spatialized WSJ0-2Mix Evaluation Results}\label{app:result_spatial}
\begin{table*}[h]
\centering
\caption{Evaluation results for 2-channel Spatialized WSJ0-2Mix. Note that the microphone positions are random for this dataset. Top results are emphasized in \colorbox{top1}{top1}, \colorbox{top2}{top2}, and \colorbox{top3}{top3}. Methods denoted with $*$ means it is impractical.}
\setlength\tabcolsep{3.3pt}
\renewcommand{\arraystretch}{1}  % Reduced line spacing
\label{tab:spatialwsj2}
\begin{tabular}{@{}cccccccccccc@{}} % Added a column for the row numbers
\toprule
 \multirow{2}{*}{Row} & \multirow{2}{*}{Methods} & \multirow{2}{*}{Unsup.} & \multirow{2}{*}{\shortstack{Array \\ Agnostic}} & \multirow{2}{*}{Prior} & \multirow{2}{*}{\shortstack{IVA \\ Init.}} & \multirow{2}{*}{\shortstack{Ref. \\ Guide.}} & \multirow{2}{*}{\shortstack{SDR \\ (dB)}} & \multirow{2}{*}{\shortstack{SI-SDR \\ (dB)}} & \multirow{2}{*}{PESQ} & \multirow{2}{*}{eSTOI} & \\
 & & & & & & & & & & & \\
\midrule
\rowcolor{gray!10}
0 & Mixture  & - & - & - & - & - & 0.2 & 0.0 & 1.81 & 0.545 & \\
\midrule
% 1a & Spatial Cluster & \checkmark & \checkmark & - & - & - & 9.3 & 8.0 & 2.48 & 0.745 & \\
1a & Spatial Cluster & \checkmark & \checkmark & - & - & - & 7.9 & 6.5 & 2.38 & 0.689 & \\
\rowcolor{gray!10}
1b & IVA & \checkmark & \checkmark & Laplace & - & - & 9.3 & 8.1 & 2.52 & 0.725 & \\
1c & IVA & \checkmark & \checkmark & Gaussian & - & - & 10.9 & 9.8 & 2.68 & 0.770 & \\
\rowcolor{gray!10}
1d & UNSSOR & \checkmark & \texttimes & - & - & - & 0.3 & -2.7 & 1.78 & 0.478 & \\
\midrule
2a & \method-A & \checkmark & \checkmark & Anechoic & \checkmark & \checkmark & 14.5$\pm$0.7 & 13.7$\pm$0.7 & 3.32$\pm0.08$ & 0.853$\pm$0.014& \\
\rowcolor{gray!10}
2b & \method-D & \checkmark & \checkmark & Anechoic & \checkmark & \checkmark & 5.6$\pm$4.0 & 3.6$\pm$4.6 & 2.34$\pm$0.45 & 0.652$\pm$0.114 & \\
\midrule
3a & \method-A-Max*& \checkmark & \checkmark & Anechoic & \checkmark & \checkmark & \cellcolor{top1}{15.3} & \cellcolor{top1}{14.6} & \cellcolor{top2}{3.41} & \cellcolor{top2}{0.870} & \\
\rowcolor{gray!10}
3b & \method-D-Max*& \checkmark & \checkmark & Anechoic & \texttimes & \checkmark & 11.2 & 9.9 & 2.96 & 0.799 & \\
\midrule
4a & \method-A-ML & \checkmark & \checkmark & Anechoic & \checkmark & \checkmark & \cellcolor{top3}{15.1}& \cellcolor{top3}{14.3}&\cellcolor{top3}{3.39} &\cellcolor{top3}{0.865}& \\
\rowcolor{gray!10}
4b & \method-D-ML & \checkmark & \checkmark & Anechoic & \texttimes & \checkmark & 10.4 & 8.9 & 2.87 & 0.773 & \\
\midrule
5a & TF-GridNet-Spatial & \texttimes & (2-mics) & - & - & - & \cellcolor{top2}{15.2} & \cellcolor{top2}{14.5} & \cellcolor{top1}{3.63} & \cellcolor{top1}{0.888} & \\
\bottomrule
\end{tabular}
\end{table*}

\begin{table*}[h]
\centering
\caption{Evaluation results for 3-channel Spatialized WSJ0-2Mix. Note that the microphone positions are random for this dataset. Top results are emphasized in \colorbox{top1}{top1}, \colorbox{top2}{top2}, and \colorbox{top3}{top3}. Methods denoted with $*$ means it is impractical.}
\setlength\tabcolsep{3.3pt}
\renewcommand{\arraystretch}{1}  % Reduced line spacing
\label{tab:spatialwsj4}
\begin{tabular}{@{}cccccccccccc@{}} % Added a column for the row numbers
\toprule
 \multirow{2}{*}{Row} & \multirow{2}{*}{Methods} & \multirow{2}{*}{Unsup.} & \multirow{2}{*}{\shortstack{Array \\ Agnostic}} & \multirow{2}{*}{Prior} & \multirow{2}{*}{\shortstack{IVA \\ Init.}} & \multirow{2}{*}{\shortstack{Ref. \\ Guide.}} & \multirow{2}{*}{\shortstack{SDR \\ (dB)}} & \multirow{2}{*}{\shortstack{SI-SDR \\ (dB)}} & \multirow{2}{*}{PESQ} & \multirow{2}{*}{eSTOI} & \\
 & & & & & & & & & & & \\
\midrule
\rowcolor{gray!10}
0 & Mixture  & - & - & - & - & - & 0.2 & 0.0 & 1.81 & 0.545 & \\
\midrule
1a & Spatial Cluster & \checkmark & \checkmark & - & - & - & 9.1 & 7.8 & 2.53 & 0.735 & \\
\rowcolor{gray!10}
1b & IVA & \checkmark & \checkmark & Laplace & - & - & 10.4 & 8.5 & 2.66 & 0.750 & \\
1c & IVA & \checkmark & \checkmark & Gaussian & - & - & 13.9 & 12.1 & 3.07 & 0.842 & \\
\rowcolor{gray!10}
1d & UNSSOR & \checkmark & \texttimes & - & - & - & 1.7 & -2.4 & 1.94 & 0.519 & \\
\midrule
2a & \method-A & \checkmark & \checkmark & Anechoic & \checkmark & \checkmark & \cellcolor{top3}{}15.7$\pm$0.6 & \cellcolor{top3}{15.0}$\pm$0.7 & 3.44$\pm$0.07 & 0.872$\pm$0.012 & \\
\rowcolor{gray!10}
2b & \method-D & \checkmark & \checkmark & Anechoic & \checkmark & \checkmark & 6.3$\pm$4.3 & 4.3$\pm$4.9 & 2.41$\pm$0.49 & 0.668$\pm$0.119 & \\
\midrule
3a & \method-A-Max*& \checkmark & \checkmark & Anechoic & \checkmark & \checkmark & \cellcolor{top1}{16.4} & \cellcolor{top1}{15.7} & \cellcolor{top2}3.52 & \cellcolor{top2}0.886 & \\
\rowcolor{gray!10}
3b & \method-D-Max*& \checkmark & \checkmark & Anechoic & \texttimes & \checkmark & 12.3 & 11.0 & 3.07 & 0.808 & \\
\midrule
4a & \method-A-ML & \checkmark & \checkmark & Anechoic & \checkmark & \checkmark & \cellcolor{top2}16.3 & \cellcolor{top2}15.5 & \cellcolor{top3}3.49 & \cellcolor{top3}0.881 & \\
\rowcolor{gray!10}
4b & \method-D-ML & \checkmark & \checkmark & Anechoic & \texttimes & \checkmark & 11.5 & 10.0 & 2.99 & 0.796 & \\
\midrule
5a & TF-GridNet-Spatial & \texttimes & (3-mics) & - & - & - & 15.5 & 14.8 & \cellcolor{top1}3.65 & \cellcolor{top1}0.892 & \\
\bottomrule
\end{tabular}
\end{table*}
In addition to the 4-channel Spatialized WSJ0-2Mix dataset's result in Table~\ref{tab:spatialwsj}, we also report 2-channel (first 2 mics) and 3-channel (first 3 mics) results on Spatialized WSJ0-2Mix dataset, in Table~\ref{tab:spatialwsj2} and Table~\ref{tab:spatialwsj4} respectively.

For the 2-channel case, we first note that in row \textbf{1d} of Table~\ref{tab:spatialwsj2}, the UNSSOR does not work properly. This is actually as expected because UNSSOR assumes an over-determined scenario~\cite {unssor}, which means the number of microphones should be larger than the number of speakers. In contrast, our \method~is able to work properly. In addition, our method easily outperforms any other unsupervised methods. Finally, we compare row \textbf{5a} with row \textbf{2a} and row \textbf{4a}, we find that the supervised method is slightly better than \method-A-ML in all metrics for the 2-channel case.

For the 3-channel results shown in Table~\ref{tab:spatialwsj4}, the result is very similar to the 4-channel results shown in Table~\ref{tab:spatialwsj}. In short, \method-A outperforms all other unsupervised methods, while UNSSOR does not work properly in this 3-channel setting. As mentioned in Sec.~\ref{app:baseline}, UNSSOR is very sensitive to ISMS loss weight $\gamma$, but we still cannot find a working weight after an extensive parameter search. Note that we ensure that our training is correct as it works for both 4-channel Spatialized WSJ0-2Mix and 3-channel SMS-WSJ datasets. One explanation is that the training dataset contains samples recorded by different 3-channel ad-hoc microphone arrays, which makes it more difficult for UNSSOR to converge. We further compare \method~(row \textbf{2a, 4a}) with the supervised baseline (row \textbf{5a}), where \method-A's mean SDR and SI-SDR can outperform the supervised TF-GridNet-Spatial.

\section{6-Channel SMS-WSJ Evaluation Results}\label{app:result_smswsj}
In addition to the 3-channel SMS-WSJ result reported in Table~\ref{tab:smswsj}, we also report a full 6-channel (all mics used) result for SMS-WSJ dataset. We first compare row \textbf{1(a-c)} with row \textbf{2a}, where our ArrayDPS-A easily outperforms spatial clustering and IVA. Comparing row \textbf{1d} with row \textbf{2d}, we can see that ArrayDPS-A's mean metric score outperforms UNSSOR by about 0.7 dB in SDR and SI-SDR, but performs a bit worse in terms of PESQ and eSTOI. However, our maximum likelihood version ArrayDPS-A-ML (row \textbf{4a}) easily outperforms UNSSOR in all metrics by a large margin. Although in 6-channel setting, \method-A-ML is still better than any other unsupervised methods, the supervised TF-GridNet-SMS in row \textbf{5a} leads our \method~by a large margin. Future research will be conducted to reduce this performance gap.
\begin{table*}[h]
\label{tab:smswsj6}
\centering
\caption{Evaluation results for 6-channel SMS-WSJ. Methods denoted with $\dagger$ are results from UNSSOR~\cite{unssor}, and methods denoted with $*$ mean it is impractical. Note that SMS-WSJ only contains samples with a fixed microphone array. Top results are emphasized in \colorbox{top1}{top1}, \colorbox{top2}{top2}, and \colorbox{top3}{top3}.}
\setlength\tabcolsep{3.3pt}
\renewcommand{\arraystretch}{1}  % Reduced line spacing
\begin{tabular}{@{}cccccccccccc@{}} % Added a column for the row numbers
\toprule
 \multirow{2}{*}{Row} & \multirow{2}{*}{Methods} & \multirow{2}{*}{Unsup.} & \multirow{2}{*}{\shortstack{Array \\ Agnostic}} & \multirow{2}{*}{Prior} & \multirow{2}{*}{\shortstack{IVA \\ Init.}} & \multirow{2}{*}{\shortstack{Ref. \\ Guide.}} & \multirow{2}{*}{\shortstack{SDR \\ (dB)}} & \multirow{2}{*}{\shortstack{SI-SDR \\ (dB)}} & \multirow{2}{*}{PESQ} & \multirow{2}{*}{eSTOI} & \\
 & & & & & & & & & & & \\
\midrule
\rowcolor{gray!10}
0 & Mixture  & - & - & - & - & - & 0.1 & 0.0 & 1.87 & 0.603 & \\
\midrule
1a & Spatial Cluster$^\dagger$ & \checkmark & \checkmark & - & - & - & 11.9 & 10.2 & 2.61 & 0.735 & \\
\rowcolor{gray!10}
1b & IVA$^\dagger$ & \checkmark & \checkmark & Laplace & - & - & 10.6 & 8.9 & 2.58 & 0.764 & \\
1c & IVA & \checkmark & \checkmark & Gaussian & - & - & 14.7 & 13.4 & 3.07 & 0.865 & \\
\rowcolor{gray!10}
1d & UNSSOR$^\dagger$ & \checkmark & \texttimes & - & - & - & 15.7 & 14.7 & 3.47 & 0.884 & \\
\midrule
2a & \method-A & \checkmark & \checkmark & Anechoic & \checkmark & \checkmark & 16.3$\pm$1.2 & 15.4$\pm$1.3 & 3.45$\pm$0.12 & 0.873$\pm$0.019 & \\
\rowcolor{gray!10}
2b & \method-D & \checkmark & \checkmark & Anechoic & \checkmark & \checkmark & 8.9$\pm$4.9 & 7.2$\pm$5.4 & 2.67$\pm$0.54 & 0.743$\pm$0.115 & \\
\midrule
3a & \method-A-Max*& \checkmark & \checkmark & Anechoic & \checkmark & \checkmark & \cellcolor{top2}{17.6} & \cellcolor{top2}{16.9} & \cellcolor{top2}{3.59} & \cellcolor{top2}{0.896} & \\
\rowcolor{gray!10}
3b & \method-D-Max*& \checkmark & \checkmark & Anechoic & \texttimes & \checkmark & 15.1 & 14.0 & 3.32 & 0.873 & \\
\midrule
4a & \method-A-ML & \checkmark & \checkmark & Anechoic & \checkmark & \checkmark & \cellcolor{top3}{17.4} & \cellcolor{top3}{16.6} & \cellcolor{top3}{3.56} & \cellcolor{top3}{0.891} & \\
\rowcolor{gray!10}
4b & \method-D-ML & \checkmark & \checkmark & Anechoic & \texttimes & \checkmark & 14.6 & 13.5 & 3.27 & 0.862 & \\
\midrule
5a & TF-GridNet-SMS$^\dagger$ & \texttimes & \texttimes & - & - & - & \cellcolor{top1}{19.4} & \cellcolor{top1}{18.9} & \cellcolor{top1}{4.08} & \cellcolor{top1}{0.949} & \\
% 5b & TF-GridNet-SMS & \texttimes & \texttimes & - & - & - & - & - & - & - & \\
\bottomrule
\end{tabular}
\end{table*}

\section{3-Speaker Evaluation Results}\label{app:3speakers}
For 3-speaker source separation, we evaluate on the spatialized WSJ0-3Mix dataset, which is the 3-speaker version of the ad-hoc spatialized WSJ0-2Mix dataset. We evaluate on the first 4 microphones, where the results are shown in Table~\ref{tab:3spk}. As shown in Row \textbf{2a-5a}, ArrayDPS outperforms the supervised baseline by a large margin in all metrics. If we sample five samples and then take the maximum likelihood one (Row \textbf{4a}), ArrayDPS outperforms supervised TF-GridNet-Spatial by more than 3 dB in terms of SDR and SI-SDR. This shows the potential of ArrayDPS in more diverse acoustic scenarios.
\begin{table*}[t]
\centering
\caption{Evaluation results for 3-speaker reverberant speech separation on 4-channel Spatialized WSJ0-3Mix. Note that the microphone positions are random for this dataset. Top results are emphasized in \colorbox{top1}{top1}, \colorbox{top2}{top2}, and \colorbox{top3}{top3}. Methods denoted with $*$ means it is impractical.}
\setlength\tabcolsep{3.3pt}
\renewcommand{\arraystretch}{1}  % Reduced line spacing
\label{tab:3spk}
\begin{tabular}{@{}cccccccccccc@{}} % Added a column for the row numbers
\toprule
 \multirow{2}{*}{Row} & \multirow{2}{*}{Methods} & \multirow{2}{*}{Unsup.} & \multirow{2}{*}{\shortstack{Array \\ Agnostic}} & \multirow{2}{*}{Prior} & \multirow{2}{*}{\shortstack{IVA \\ Init.}} & \multirow{2}{*}{\shortstack{Ref. \\ Guide.}} & \multirow{2}{*}{\shortstack{SDR \\ (dB)}} & \multirow{2}{*}{\shortstack{SI-SDR \\ (dB)}} & \multirow{2}{*}{PESQ} & \multirow{2}{*}{eSTOI} & \\
 & & & & & & & & & & & \\
\midrule
\rowcolor{gray!10}
0 & Mixture  & - & - & - & - & - & -3.0 & -3.3 & 1.50 & 0.371 & \\
\midrule
1a & Spatial Cluster & \checkmark & \checkmark & - & - & - & 6.9 & 5.6 & 2.10 & 0.625 & \\
\rowcolor{gray!10}
1b & IVA & \checkmark & \checkmark & Laplace & - & - & 5.9 & 3.8 & 2.10 & 0.574 & \\
1c & IVA & \checkmark & \checkmark & Gaussian & - & - & 9.6 & 7.6 & 2.46 & 0.701 & \\
\rowcolor{gray!10}
1d & UNSSOR & \checkmark & \texttimes & - & - & - & -2.8 & -6.1 & 1.50 & 0.313 & \\
\midrule
2a & \method-A & \checkmark & \checkmark & Anechoic & \checkmark & \checkmark & \cellcolor{top2}{12.8}$\pm$1.2 & \cellcolor{top2}{11.8}$\pm$1.3 & \cellcolor{top2}{3.11$\pm$0.16} & \cellcolor{top2}{0.791$\pm$0.028} & \\
\rowcolor{gray!10}
3a & \method-A-Max$^*$ & \checkmark & \checkmark & Anechoic & \checkmark & \checkmark & {14.3} & {13.4} & 3.31 & 0.816 & \\
% \rowcolor{gray!10}
4a & \method-A-ML & \checkmark & \checkmark & Anechoic & \checkmark & \checkmark & \cellcolor{top1}14.0 & \cellcolor{top1}13.0 & \cellcolor{top1}3.27 & \cellcolor{top1}0.807 & \\
\rowcolor{gray!10}
5a & TF-GridNet-Spatial & \texttimes & (4-mics) & - & - & - & \cellcolor{top3}10.6 & \cellcolor{top3}9.7 & \cellcolor{top3}2.85 & \cellcolor{top3}0.747 & \\
\bottomrule
\end{tabular}
\end{table*}

\section{Ablations on Sensitivity to Hyperparameters}\label{app:ablations}
To evaluate our ArrayDPS's sensitivity and robustness to hyper-parameters, we conduct ablation experiments on the likelihood score guidance $\xi$, the sampling starting noise level $\tau_{\text{max}}$, FCP filter length $F$, the number of steps $N_{fg}$ to use the IVA initialized filters, whether to use Gaussian or Laplace prior for IVA initialization, and whether to use the IVA separated source as diffusion initialization (line 8 of Algorithm~\ref{alg:ubliss}). For each parameter, we do ablations by modifying it starting from the default parameter mentioned in Appendix~\ref{app:baseline}. All the ablations are conducted on the first 50 samples of the SMS-WSJ validation set. Results are shown in Table~\ref{tab:ablation-xi}-\ref{tab:ablation-warm}
.
% ----------- ξ Ablation -----------
\begin{table}[h]
    \centering
    \caption{Ablation study for $\xi$.}
    \label{tab:ablation-xi}
    \begin{tabular}{lccccccccc}
        \toprule
        $\xi$ & 1.2 & 1.4 & 1.6 & 1.8 & 2.0 & 2.2 & 2.4 & 2.6 & 2.8 \\
        \midrule
        SI-SDR & 15.2$\pm$1.1 & 15.5$\pm$1.1 & 15.6$\pm$1.0 & 15.6$\pm$1.2 & 
        15.3$\pm$1.6 & 15.0$\pm$1.5 & 15.0$\pm$1.6 & 15.2$\pm$2.6 & 15.1$\pm$1.8 \\
        \bottomrule
    \end{tabular}
\end{table}

% ----------- τ_max Ablation -----------
\begin{table}[H]
    \centering
    \caption{Ablation study for $\tau_{\text{max}}$.}
    \label{tab:ablation-taumax}
    \begin{tabular}{lccccccc}
        \toprule
        $\tau_{\text{max}}$ & 0.5 & 0.6 & 0.7 & 0.8 & 0.9 & 1.0 & 1.1 \\
        \midrule
        SI-SDR & 16.1$\pm$0.9 & 16.0$\pm$0.8 & 15.8$\pm$1.2 & 15.6$\pm$1.1 & 
        15.1$\pm$1.7 & 14.8$\pm$1.8 & 14.3$\pm$2.1 \\
        \bottomrule
    \end{tabular}
\end{table}

% ----------- F Ablation -----------
\begin{table}[H]
    \centering
    \caption{Ablation study for FCP filter length $F$.}
    \label{tab:ablation-F}
    \begin{tabular}{lccccccc}
        \toprule
        $F$ & 10 & 11 & 12 & 13 & 14 & 15 & 16 \\
        \midrule
        SI-SDR & 15.8$\pm$1.3 & 15.5$\pm$1.6 & 15.6$\pm$1.2 & 15.1$\pm$1.7 & 
        14.9$\pm$1.7 & 14.9$\pm$1.5 & 14.8$\pm$1.6 \\
        \bottomrule
    \end{tabular}
\end{table}

% ----------- N_fg Ablation -----------
\begin{table}[H]
    \centering
    \caption{Ablation study for filter guidance steps $N_{fg}$.}
    \label{tab:ablation-Nfg}
    \begin{tabular}{lccccccc}
        \toprule
        $N_{fg}$ & 70 & 80 & 90 & 100 & 110 & 120 & 130 \\
        \midrule
        SI-SDR & 15.0$\pm$1.7 & 15.2$\pm$1.7 & 14.6$\pm$2.1 & 15.2$\pm$1.7 & 
        15.4$\pm$1.2 & 15.5$\pm$1.1 & 15.5$\pm$1.1 \\
        \bottomrule
    \end{tabular}
\end{table}

% ----------- IVA Prior Type Ablation -----------
\begin{table}[H]
    \centering
    \caption{Ablation study for IVA prior used for IVA initialization.}
    \label{tab:ablation-iva}
    \begin{tabular}{lcc}
        \toprule
        IVA Prior & Gaussian & Laplace \\
        \midrule
        SI-SDR & 15.3$\pm$1.6 & 14.8$\pm$1.6 \\
        \bottomrule
    \end{tabular}
\end{table}

% ----------- Warm Initialization Ablation -----------
\begin{table}[H]
    \centering
    \caption{Ablation study for using IVA separated sources as diffusion initialization (line 8 of Algorithm~\ref{alg:ubliss}).}
    \label{tab:ablation-warm}
    \begin{tabular}{lcc}
        \toprule
        Diffusion Initialization& Yes & No \\
        \midrule
        SI-SDR & 15.3$\pm$1.6 & 13.3$\pm$2.3 \\
        \bottomrule
    \end{tabular}
\end{table}

\section{Filter and Source Visualization}\label{app:visualization}
Figure~\ref{fig:sample1}-~\ref{fig:sample4} gives visualization of the separated sources and the final FCP estimated relative RIRs. 
In each figure of Figure~\ref{fig:sample1}-~\ref{fig:sample4}, the first row contains the ground-truth anechoic source 1 $\Tilde{s}_1$, the ground-truth RIR $h_{1,1}$ (from anechoic source to reference channel), and the reference-channel reverberant source 1 $s_{1,1}$. Obviously, the third signal is the convolution of the first two signals. Then the second row shows the same for our ArrayDPS, i.e., separated virtual source 1 $\hat{s}_0$, the final FCP estimated relative RIR $\hat{g}_{0\rightarrow 1}^1$, and the separated reference-channel source 1 $\hat{s}_{1,1}$. These two adjacent rows allow direct comparison between the ground-truth and ArrayDPS. Similarly, row 3 and row 4 show results for source 2.

These figures show us the difference between the final FCP estimated filter and the ground truth RIR. Those two signals might not be aligned, but sometimes show similar structure on the spectrogram. On the other hand, the virtual source separated by ArrayDPS is more like the anechoic signal rather than the reverberant signal. We explain this for two reasons: 1) the diffusion model is trained mostly on clean anechoic speech, hence the inclination to generate anechoic sources; 2) FCP performs much better when the input source signal is anechoic (see Appendix~\ref{app:fcp}). This means outputting an anechoic source (instead of a reverberant one) would satisfy a higher likelihood.

Our demo site also shows virtual sources with 2-speaker separation. We urge readers to listen to the virtual source samples and compare them with the final output and the ground-truth anechoic source. It appears that ArrayDPS is indeed accomplishing some dereverberation, however, we do not want to over-claim here and intend to verify this promise in future research.

\begin{figure}[H]
    \centering
    \includegraphics[width=0.7\linewidth]{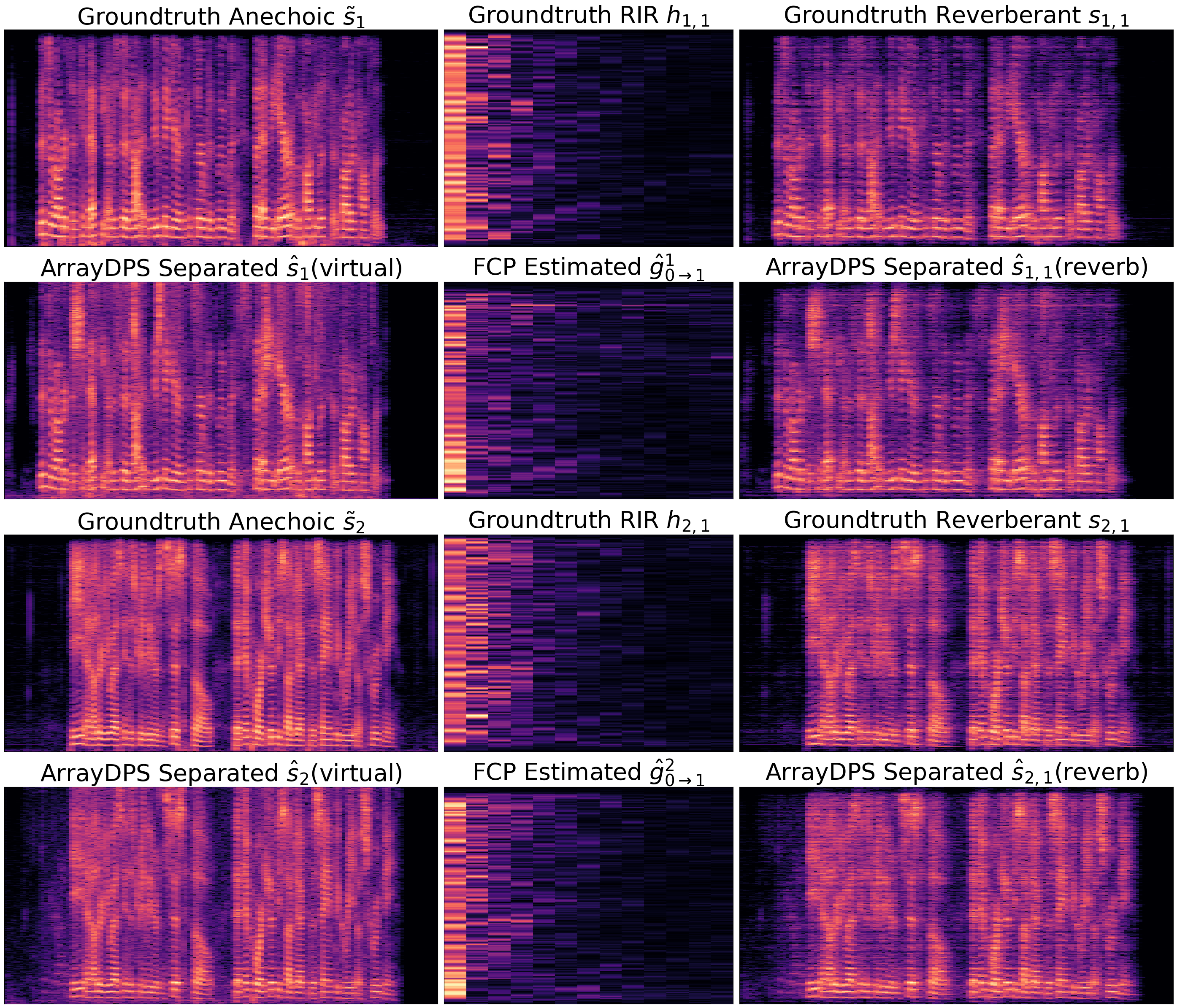}
    \vspace{-10pt}
    \caption{SMS-WSJ Sample 0\_442c040o\_443c040g Visualization}
    \label{fig:sample1}
\end{figure}
\begin{figure}[H]
    \centering
    \includegraphics[width=0.7\linewidth]{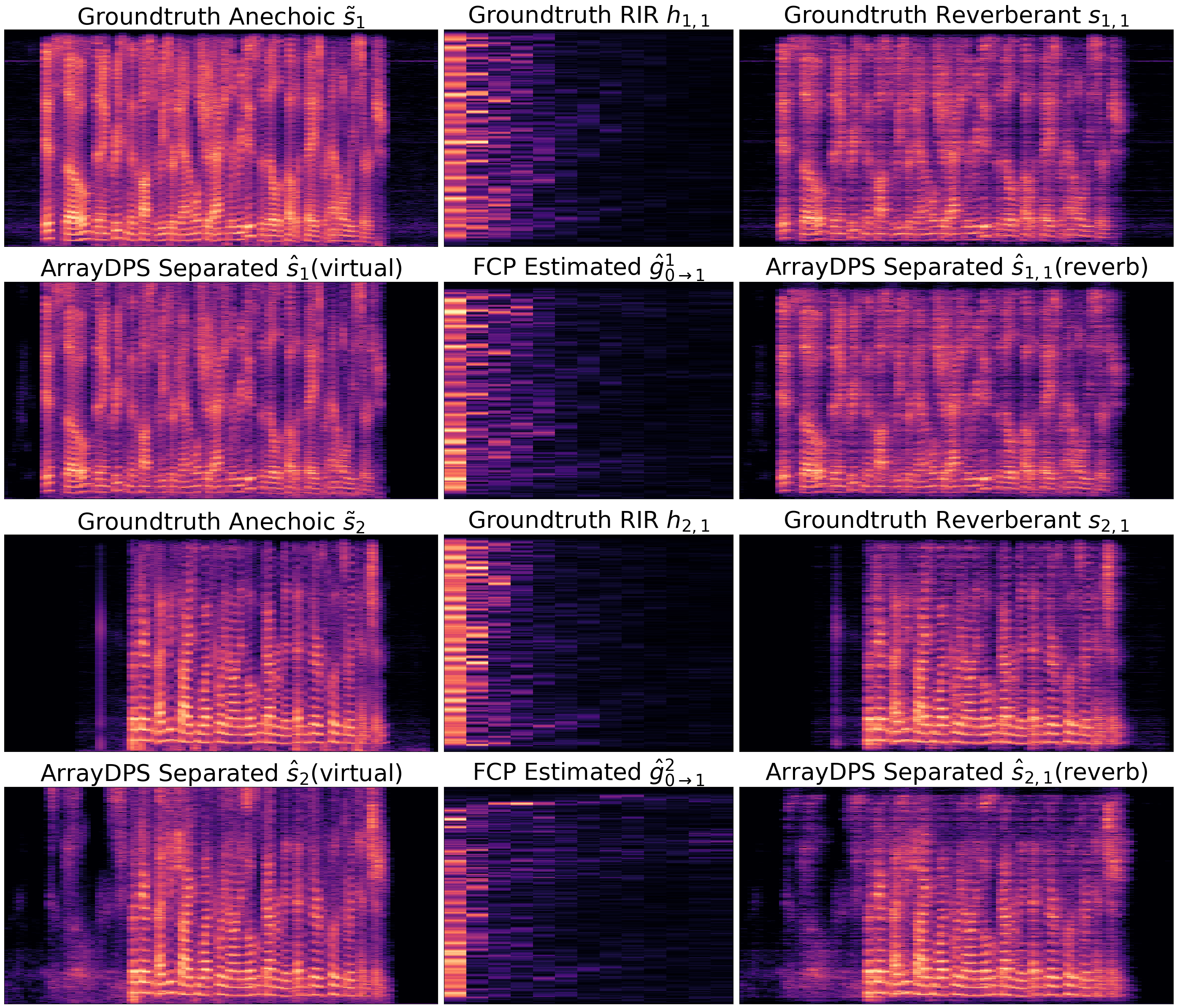}
    \vspace{-10pt}
    \caption{SMS-WSJ Sample 1015\_446c0415\_442c040c Visualization}
    \label{fig:sample2}
\end{figure}
\begin{figure}[H]
    \centering
    \includegraphics[width=0.7\linewidth]{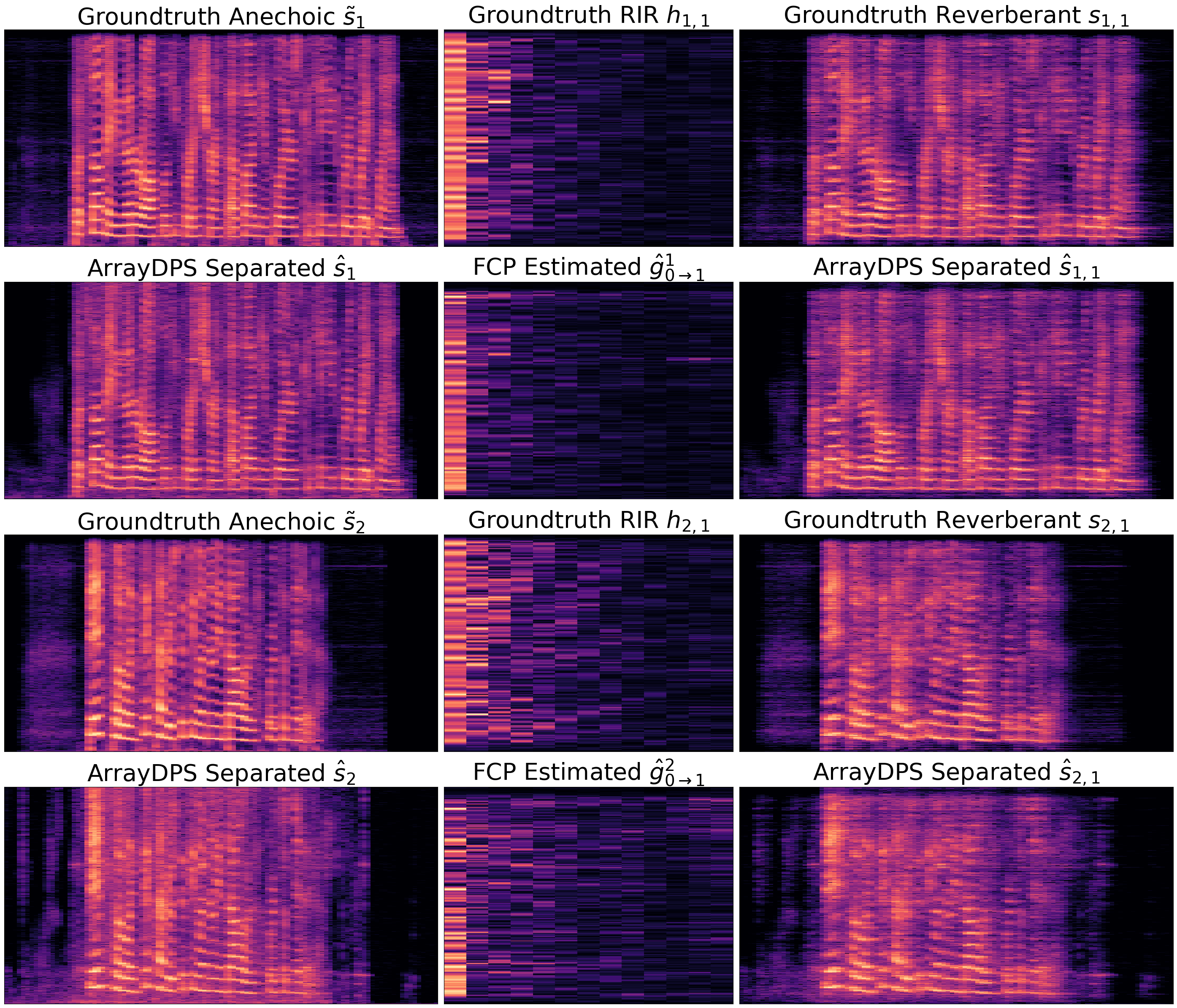}
    \vspace{-10pt}
    \caption{SMS-WSJ Sample 1120\_445c040c\_441c040m Visualization}
    \label{fig:sample3}
\end{figure}
\begin{figure}[H]
    \centering
    \includegraphics[width=0.7\linewidth]{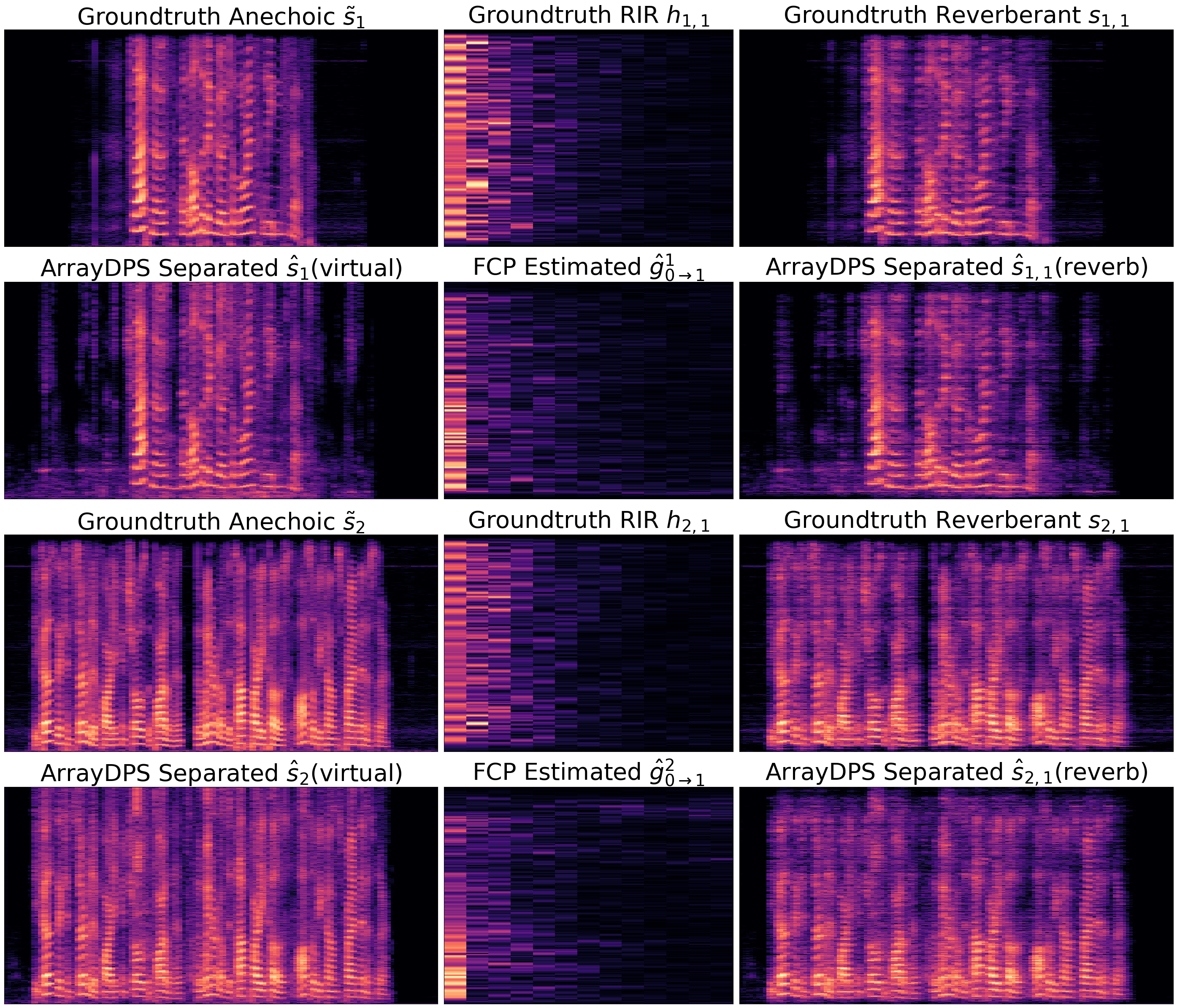}
    \vspace{-10pt}
    \caption{SMS-WSJ Sample 999\_441c040c\_447c040k Visualization}
    \label{fig:sample4}
\end{figure}
% \section{Dereverberation Effect of \method}\label{app:dereverb}

% \section{Relative RIR visualization}\label{app:visualization}
% You can have as much text here as you want. The main body must be at most $8$ pages long.
% For the final version, one more page can be added.
% If you want, you can use an appendix like this one.  

% The $\mathtt{\backslash onecolumn}$ command above can be kept in place if you prefer a one-column appendix, or can be removed if you prefer a two-column appendix.  Apart from this possible change, the style (font size, spacing, margins, page numbering, etc.) should be kept the same as the main body.
%%%%%%%%%%%%%%%%%%%%%%%%%%%%%%%%%%%%%%%%%%%%%%%%%%%%%%%%%%%%%%%%%%%%%%%%%%%%%%%
%%%%%%%%%%%%%%%%%%%%%%%%%%%%%%%%%%%%%%%%%%%%%%%%%%%%%%%%%%%%%%%%%%%%%%%%%%%%%%%

\end{document}